\documentclass[12pt]{amsart}
\usepackage{amssymb,amscd}
\usepackage{verbatim}

\usepackage[colorlinks=true,allcolors = blue]{hyperref} 

\let\frak\mathfrak

\def\>{\relax\ifmmode\mskip.666667\thinmuskip\relax\else\kern.111111em\fi}
\def\<{\relax\ifmmode\mskip-.333333\thinmuskip\relax\else\kern-.0555556em\fi}
\def\vsk#1>{\vskip#1\baselineskip}
\def\vv#1>{\vadjust{\vsk#1>}\ignorespaces}
\def\vvn#1>{\vadjust{\nobreak\vsk#1>\nobreak}\ignorespaces}

\let\Medskip\medskip
\def\medskip{\par\Medskip}
\let\Bigskip\bigskip
\def\bigskip{\par\Bigskip}

\let\Maketitle\maketitle
\def\maketitle{\Maketitle\thispagestyle{empty}\let\maketitle\empty}

\newtheorem{thm}{Theorem}[section]
\newtheorem{cor}[thm]{Corollary}
\newtheorem{lem}[thm]{Lemma}

\newtheorem{defn}[thm]{Definition}

\numberwithin{equation}{section}

\theoremstyle{definition}
\newtheorem*{rem}{Remark}

\let\mc\mathcal
\let\nc\newcommand

\let\al\alpha

\let\la\lambda
\let\La\Lambda

\let\phi\varphi
\let\si\sigma

\let\Om\Omega

\let\der\partial

\let\ox\otimes

\let\leq\leqslant

\let\on\operatorname
\let\bi\bibitem
\let\bs\boldsymbol

\def\C{{\mathbb C}}
\def\Z{{\mathbb Z}}
\def\P{{\mathbb P}}

\def\B{{\mc B}}
\def\D{{\mc D}}

\def\F{{\mc F}}
\def\L{{\mc L}}

\def\+#1{^{\{#1\}}}

\def\End{\on{End}}

\def\Res{\on{Res}}

\def\rdet{\on{rdet}}

\def\Wr{\on{Wr}}

\def\gln{{\mathfrak{gl}_N}}
\def\gl{\mathfrak{gl}}
\def\sln{\mathfrak{sl}_N}
\def\sl{\mathfrak{sl}}

\def\Ugln{U(\gln)}

\def\beq{\begin{equation}}
\def\eeq{\end{equation}}
\def\be{\begin{equation*}}
\def\ee{\end{equation*}}

\nc{\bea}{\begin{eqnarray*}}
\nc{\eea}{\end{eqnarray*}}
\nc{\bean}{\begin{eqnarray}}
\nc{\eean}{\end{eqnarray}}
\nc{\Ref}[1]{{\rm(\ref{#1})}}

\def\h{{\mathfrak h}}

\nc{\Il}{{\mc I_{\bs\la}}}
\nc{\bla}{{\bs\la}}
\nc{\Fla}{\F_{\bs\la}}
\nc{\tfl}{{T^*\Fla}}
\nc{\GL}{{GL_n(\C)}}
\nc{\GLC}{{GL_n(\C)\times\C^*}}

\def\ka{{\kappa}}
\def\slt{{\frak{sl}_2}}

\nc{\glt}{{\frak{gl}_2}}
\def\KZ/{{\slshape KZ\/}}
\def\qKZ/{{\slshape qKZ\/}}
\def\XXX/{{\slshape XXX\/}}
\nc{\arr}{\rightarrow}
\nc{\larr}{\longrightarrow}
\nc{\A}{{\mc A}}
\nc{\Ax}{{\mc A(\xi)}}
\nc{\cdet}{{\on{cdet}\,}}

\def\ep{{\epsilon}}

\def\Fun{{\on{Fun}_{\sln}(V[0])}}

\begin{document}

\hrule width0pt
\vsk->

\hrule width0pt
\vsk->

\title[Bethe algebra,  eigenfunctions, and theta-polynomials]
{Dynamical elliptic Bethe algebra, \\
 KZB eigenfunctions, and theta-polynomials}

\author[D.\,Thompson and  A.\,Varchenko]
{Daniel Thompson$\>^\circ$  and Alexander Varchenko$\>^\star$}

\maketitle

\begin{center}
{\it $^{\circ, \star}\<$Department of Mathematics, University
of North Carolina at Chapel Hill\\ Chapel Hill, NC 27599-3250, USA\/}

\vsk.5>
{\it $^{\star}\<$Faculty of Mathematics and Mechanics, Lomonosov Moscow State
University\\ Leninskiye Gory 1, 119991 Moscow GSP-1, Russia\/}

\end{center}

\vsk>
{\leftskip3pc \rightskip\leftskip \parindent0pt \Small
{\it Key words\/}: Elliptic dynamical Bethe algebra, Bethe ansatz, theta-polynomials, \\
\phantom{aaaaaaaaaa}
 analytic and Weyl group involutions
\vsk.6>
{\it 2010 Mathematics Subject Classification\/}: 81R50 (17B37, 33C67, 14H70)
\par}

{\let\thefootnote\relax
\footnotetext{\vsk-.8>\noindent
$^\circ\<$
{\it E\>-mail}: dthomp@email.unc.edu
\\
$^\star\<$
{\it E\>-mail}:
anv@email.unc.edu\>,
supported in part by NSF grant DMS-1665239}}

\maketitle

\begin{abstract}
  Let $(\ox_{j=1}^nV_j)[0]$ be the zero weight subspace of a tensor product of finite-dimensional irreducible
$\slt$-modules. 
 The dynamical elliptic Bethe algebra is a commutative algebra of differential operators acting on $(\ox_{j=1}^nV_j)[0]$-valued
 functions on the Cartan subalgebra of $\slt$.  The algebra is generated by values of the coefficient $S_2(x)$ of a certain 
differential operator $\D=\der_x^2+S_2(x)$, defined by V.\,Rubtsov, A.\,Silantyev, D.\,Talalaev in 2009.   We express $S_2(x)$ in terms of the KZB operators introduced by
G.\,Felder and C.\,Wieszer\-kowski in 1994.  We study the eigenfunctions of the 
dynamical elliptic Bethe algebra by the Bethe ansatz method. Under certain assumptions
we show that such Bethe eigenfunctions are in one-to-one correspondence with
ordered pairs of theta-polynomials
of certain degree. The correspondence between Bethe eigenfunctions and two-dimensional  spaces, generated by 
the two theta-polynomials, 
is an analog of the non-dynamical non-elliptic correspondence between the eigenvectors of the $\glt$ Gaudin model and
the two-dimensional subspaces of the vector space $\C[x]$, due to
 E.\,Mukhin, V.\,Tarasov, A.\,Varchenko. 

We obtain a counting result for, equivalently, certain solutions of the Bethe ansatz equation, certain fibers of the 
elliptic Wronski map, or ratios of theta polynomials, whose derivative is of a certain form.
  We give an asymptotic expansion for Bethe eigenfunctions in a certain limit, and deduce from that 
that the Weyl involution acting on Bethe eigenfunctions 
coincides with the action of an analytic involution given by the transposition of
theta-polynomials in the associated ordered pair.
\end{abstract}

{\small\tableofcontents\par}

\section{Introduction}

Let  $V=\otimes_{j=1}^n V_j$ be a tensor product of $\sln$-modules, $V[0]$ the zero weight subspace,
 $z_1,\dots,z_n, \tau\in \C$ with $\on{Im}\tau>0$.
  Let $\on{Fun}_{\sln}(V[0])$  be the space of $V[0]$-valued functions
 on the Cartan subalgebra of  $\sln$. 
Given these data we follow  V.\,Rubtsov, A.\,Silantyev, D.\,Talalaev  in \cite{RST} and introduce
a  commutative algebra of linear differential operators acting on $\on{Fun}_{\sln}(V[0])$,
which we call the {\it dynamical elliptic Bethe algebra}. \!\!\footnote{In fact in \cite{RST} 
the authors consider the tensor products of $\gln$-modules, 
while in this paper we restrict ourselves to tensor products of $\sln$-modules}
 The algebra is given in the form of a differential operator
\bean
\label{UDo}
\D= \der^N_x + \sum_{j=1}^N S_j(x) \der_x^{N-j} 
\eean
with respect to the variable $x$, where each coefficient $S_j(x)$ is a differential operator
acting  on 
$\on{Fun}_{\sln}(V[0])$ and depending on  $x$ as a parameter. We call the operator $\D$ the
{\it universal differential operator} or the {\it RST-operator}. 
The dynamical elliptic Bethe algebra $\B^V(z_1,\dots,z_n,\tau)$
is the algebra generated by the identity operator and  the
operators $\{S_j(x)\ |\  x\in \C, j=1,\dots,N\}$.

Our Theorem \ref{thm trans} says that the Bethe algebra is doubly periodic,
\bea
\D(x+ k+ l \tau) = \D(x), \qquad k,l\in \Z,
\eea
Theorem \ref{thm onz} says that the Bethe algebra is conjugated by some explicit operator, if
some of the complex numbers $z_1,\dots,z_n$ are shifted by elements of the lattice $\Z+\tau\Z$.

\medskip
Let $\Psi \in \on{Fun}_{\sln}(V[0])$ be an eigenfunction of the Bethe algebra. Then
\bea
S_j(x) \Psi = B_j(x) \Psi,\qquad j=1,\dots,N,
\eea
where $B_j(x)$ is a scalar function of $x$ of eigenvalues of the operator $S_j(x)$.
 This construction assigns
 to $\Psi$ a scalar  differential operator with respect to the variable $x$,
\bean
\label{fund PsiN}
\D_{\Psi} = \der^N_x + \sum_{j=1}^N B_j(x) \der_x^{N-j}, 
\eean
called the {\it  fundamental differential operator} of the eigenfunction $\Psi$. We have
$\D_\Psi(x+ k+ l \tau) = \D_\Psi(x)$ for $k,l\in \Z$.

 The correspondence 
  $\Phi \to \D_\Psi$ between eigenfunctions of the commutative dynamical elliptic
 Bethe algebra and
ordinary differential operators of special form
 is in the spirit of the geometric Langlands correspondence. An example  of a non-dynamical
non-elliptic 
correspondence of that type see in \cite{MTV1, MTV2}, where an  eigenvector $\Psi$
of the non-dynamical $\gln$ Gaudin model  corresponds
to a differential operator, whose kernel consists of polynomials only, and 
that polynomial kernel defines a point in the Grassmannian of $N$-planes in the 
vector space $\C[x]$. 

\smallskip

In this paper we study the relation  $\Phi \to \D_\Psi$ for the Bethe eigenfunctions
in the case of the tensor product of irreducible $\slt$-modules. In that case
\bea
\D=\der_x^2 + S_2(x) \qquad\on{and}\qquad  \D_\Psi=\der_x^2 + B_2(x).
\eea

In Lemma \ref{weyl inv} we show that the universal differential operator
$\der_x^2 + S_2(x) $
is invariant with respect to the natural action of the $\slt$ Weyl group $W=\{\on{id}, s\}$. 
  Hence the Weyl group acts on the set of eigenfunctions
of the dynamical elliptic Bethe algebra, and the eigenfunctions
$\Psi$ and $s(\Psi)$ have the same fundamental differential operators.

In Theorem \ref{prop:S_2}  we show that
\bean
\label{2S}
S_2(x,z,\tau) = -2\pi i H_0(z,\tau) - \sum_{s=1}^n \Big[H_s(z,\tau)\rho(x-z_s,\tau) + c_2^{(s)} \rho'(x-z_s,\tau) \Big],
\eean
where  $H_j(z,\tau)$, $j=0,\dots,n$, are the KZB operators, introduced
 by G.\,Felder and C.\,Wieszer\-kowski in \cite{FW}
to describe the differential equations for conformal blocks on the elliptic curve $\C/\Z+\tau\Z$. The operator $c^2$ 
in \Ref{2S} is the 
quadratic central element in $U(\slt)$, and $\rho =\theta\rq{}/\theta$ is the logarithmic derivative 
of the normalized first Jacobi theta function, see \Ref{1nt}.

Formula \Ref{2S} shows that the dynamical elliptic Bethe algebra is generated by the KZB operators for the Lie algebra $\slt$. 

\medskip
The construction of eigenfunctions of the KZB operators by the Bethe ansatz method
 was suggested in \cite{FV1}.
Let $\la_{12} $ be the coordinate on the Cartan subalgebra of $\slt$. 
Let $V=\otimes_{j=1}^n V_j$
be a tensor product of irreducible $\slt$-modules with highest weights $m_j$, $j=1,\dots,n$. 
The zero weight subspace $V[0]$ in nontrivial if the sum $m_1+\dots+m_n$ is even, 
that is, the sum equals $2m$ for some $m\in\Z_{>0}$.

We introduces a certain {\it elliptic master function}
$\Phi(\mu,t,z,\tau)$, depending on complex variables $\mu, t=(t_1,\dots,t_m), z=(z_1,\dots,z_n)$, $\tau$,
and introduce a certain  meromorphic  $V[0]$-valued function $\Psi(\la_{12},\mu,t,z,\tau)$ of $\la_{12}$,
 depending on the parameters $t, z,\mu,\tau$, 
see \Ref{master} and \Ref{eq eig}. 

Given $\mu, z,\tau$, we consider  the critical point equations for the master function with respect to the variables
$t$,
\bea
\frac{\der \Phi}{\der t_j}(\mu,t,z,\tau) =0, \qquad j=1,\dots,m,
\eea
called the {\it Bethe ansatz equations.} By \cite{FV1}, if $(\mu,t,z,\tau)$ is a solutions of the Bethe ansats equations,
then the function $\Psi(\la_{12},\mu,t,z,\tau)$ of $\la_{12}$ is an eigenfunctions of the KZB operators,
\bean
\label{KZb}
H_0(z,\tau) \Psi(\la_{12}, \mu, t,  z,\tau) 
&=&
 \frac {\der \Phi}{\der \tau}(\mu, t, z,\tau)\, \Psi(\la_{12}, \mu, t,  z,\tau),
\\
\notag
H_a(z,\tau)\, \Psi(\la_{12}, \mu, t,  z,\tau) 
&=&
 \frac {\der \Phi}{\der z_a}(\mu, t, z,\tau) \,\Psi(\la_{12}, \mu, t,  z,\tau), \qquad a=1,\dots,n,
\eean
see Theorem \ref{FV thm}.  The eigenfunction $\Psi(\la_{12}, \mu, t,  z,\tau)$ has a periodicity property:
\bea
\Psi(\la_{12}+1, \mu, t,  z,\tau)= e^{\pi i \mu} \Psi(\la_{12}, \mu, t,  z,\tau).
\eea

Thus, given $z,\tau$ we have a family of meromorphic eigenfunctions of the 
elliptic dynamical Bethe
algebra depending on the parameters $(\mu, t)$ such that $(\mu, t,z,\tau)$ solve the Bethe ansatz equations.

\medskip

The fundamental differential operator of the eigenfunction $\Psi(\la_{12}, \mu, t,  z,\tau)$ is denoted by
$\D_{(\mu,t,z,\tau)}$.  Formulas \Ref{2S} and \Ref{KZb} allow us to give the following formula for
$\D_{(\mu,t,z,\tau)}$.  Let 
\bea
y(x) ={\prod}_{i=1}^m\theta(x-t_i,\tau),
\qquad
u(x)= e^{\pi i \mu x}y(x){\prod}_{s=1}^n\theta(x-z_s,\tau)^{-m_s/2}.
\eea
Then
\bean
\label{d_t}
\D_{(\mu, t, z,\tau)} =\big(\der_x + (\ln u)\rq{}\big)\big(\der_x - (\ln u)\rq{}\big),
\eean
where ${}\rq{}$ is the derivative with respect to $x$, 
see Theorem \ref{prop fund diff op}. 

\medskip
Having this formula, we restrict ourselves to  the case $V = (\C^2)^{\ox 2m}$, study the kernels of the 
fundamental differential operators  $\D_{(\mu,t,z,\tau)}$, and the dependence of the kernels on
$(\mu,t)$.

\medskip
We introduce the space of theta-polynomials of degree $m$ and show, under certain conditions, that
the kernel of $\D_{(\mu,t,z,\tau)}$ is generated by two functions $u_1(x), u_2(x)$ of the form,
\bean
\label{uu12}
u_1(x) =f/\sqrt{\Wr(f,g)},
\qquad
 u_2(x) = g/\sqrt{\Wr(f,g)},
\eean
where $f,g$ are theta-polynomials of degree $m$, and
$\Wr(f,g) = fg\rq{}-f\rq{}g$ is the Wronskian of the functions $f$ and $g$,
see  Theorems \ref{lem:sum} and \ref{lem fo=}.

Under certain conditions, we show that the  Bethe eigenfunctions
$\Psi(\la_{12}$, $\mu, t,  z,\tau)$ are in one-to-one correspondence with 
equivalence classes of ordered pairs
$(f,g)$ of theta-polynomi\-als of degree $m$, see Theorems \ref{lem:sum}, \ref{lem equ},
and \ref{thm BdE}.

\medskip

We define the {\it analytic involution} on the set of ordered
pairs of theta-polynomials by the formula $(f,g) \mapsto (g,f)$.
 The analytic involution induces an involution on the set of Bethe eigenfunctions. 
We prove that the analytic involution on the set of Bethe eigenfunctions coincides with the action of the
Weyl involution $s$,   see Theorem \ref{thm inv}, see also \cite[Section 5]{MV2}, where a non-dynamical
non-elliptic version of this result had been established.

\medskip
Under certain conditions, we
  show that  the differential operators $\D_{(\mu,t,z,\tau)}$ are in one-to-one correspondence with
the unordered pairs  $(\Psi, s(\Psi))$ of Bethe eigenfunctions, see Theorem \ref{thm BdE}. That means 
that two Bethe eigenfunctions  $\Psi(\la_{12},\mu,t,z,\tau)$ and
 $\Psi(\la_{12},\mu\rq{},t\rq{},z,\tau)$ have the same eigenvalues for every element of the dynamical 
elliptic Bethe algebra $\B^V(z_1$, \dots, $z_{2m},\tau)$, if and only if 
either $\Psi(\la_{12},\mu,t,z,\tau)=\Psi(\la_{12},\mu\rq{},t\rq{},z,\tau)$ or $s(\Psi(\la_{12},\mu,t,z,\tau))=
\Psi(\la_{12},\mu\rq{}$, $t\rq{},z,\tau)$ up to proportionality.

\medskip

In the proofs we use asymptotics of the solutions of  Bethe ansatz equations when $\mu\to\infty$,
see Theorems \ref{thm deg}, \ref{cor de}, Lemma \ref{lem 1.23}. 

\medskip

As a byproduct of  that asymptotic study we obtain some counting results.
For example, in  Theorem \ref{cor de}, 
we show that in the limit $\mu\to\infty$, the number of Bethe eingenfunctions with given $\mu, z,\tau$
equals 
$\dim (\C^2)^{\ox 2m}[0] = \binom{2m}{m}$,
 and the leading terms of asymptotics  of those Bethe eigenfunctions form
the standard weight basis of $(\C^2)^{\ox 2m}[0]$.

\medskip

We also count the number of ratios of theta-polynomials of degree $m$ with prescribed zeros of the derivative.  More precisely,
fix a fundamental parallelogram $\La\subset \C$ of the lattice $\Z+\tau\Z$ acting on $\C$.
Consider  the ratio $F$ of two theta-polynomials of degree $m$ 
with $m$ simple poles in $\La$.
Then the function $F$ can be  written uniquely in the form 
\bean
\label{hfunn}
F=g/f,\qquad 
f={\prod}_{j=1}^m\theta(x-t_j,\tau),  
\eean
where $t=(t_1,\dots,t_m)$ is a point of $\La^m/S_m$ with distinct coordinates
and $g$ is a theta-polynomial of degree $m$. The derivative $F\rq{} = \Wr(f,g)/f^2$
 can be written uniquely in the form
\bean
\label{derhn}
F\rq{} = e^{-2\pi i \mu x}\frac{\prod_{a=1}^{2m}\theta(x-z_a,\tau)}{\prod_{j=1}^m\theta(x-t_j,\tau)^2}
\eean
for some $\mu\in\C$ and $z=(z_1,\dots,z_{2m})\in\La^{2m}/S_{2m}$. By our
assumption
 the numerator and denominator of this ratio have no common zeros.
Let $\La\rq{}$ be the interior of $\La$. Assume that 
$z=(z_1,\dots,z_{2m})$ belongs $(\La\rq{})^{2m}/S_{2m}$ and has distinct coordinates.
Then there exists
$N>0$, such that for any $\mu\in \C$ with $|\mu|> N$
there exist exactly $\binom{2m}m$ functions $F(x)$ as in \Ref{hfunn} with the derivative as in
\Ref{derhn}, up to  proportionality, see Theorem \ref{thm ratio}.

\medskip
The paper is organized as follows.
 In Section \ref{sec:prelim} we collect useful formulas on theta functions. 
In Section \ref{sec RST-o}  the RST-operator and KZB operators are introduced.
First properties of the RST-operator are discussed in Section  \ref{sec:first results}.
In particular, we prove the double periodicity of the RST-operator. We describe the  transformations
of the RST-operator
under the shifts of the parameters $z_1,\dots,z_n$ by
 elements of the lattice $\Z+\tau\Z$. We describe the Laurent
expansion of the RST-operator in Theorem \ref{lem laurent} and express the coefficient $S_2(x)$ in terms of the KZB operators.
The Bethe ansatz for $\slt$ is discussed in Section \ref{sec:D}.
We restrict ourselves to the case $ V=(\C^2)^{\ox 2m}$
in Section \ref{sec 6}.  In Section \ref{sec:wron} the theta-polynomials are introduced. 
In Section \ref{sec Wrd} we study the Wronskian equation $\Wr(f,g)=h$
 for theta-polynomails $f, g, h$ 
of degrees $m, m, 2m$, respectively. In particular, we present a theorem from an unpublished paper \cite{BMV} by 
L.\,Borisov, E.\,Mukhin, A.\,Varchenko, which  says that given $f, h$, then, under certain conditions,
there exists a unique $g$ satisfying the equation $\Wr(f,g)=h$, see Theorem \ref{BMV}.
In Section \ref{sec:eigenfunctions} the interrelations between solutions of the Bethe 
ansatz equations and ordered pairs of theta-polynomials are studied. 
In Section \ref{sec BAf}  a formula for Bethe eigenfunctions is given and the properties of Bethe eigenfunctions are discussed.
In Section \ref{sec elWm} we introduce the elliptic Wronki map and describe the asymptotic behavior of its fibers in 
the limit
$\on{Im}\mu\to \infty$, see
the important for us Theorem \ref{thm deg}. We discuss the applications of Theorem \ref{thm deg}
in Section \ref{sec appls}, in particular, we give a proof of
 Theorem  \ref{thm inv} that the analytic and Weyl involutions coincide
on Bethe eigenfunctions.

\medskip

The second author thanks
 A.\,Eremenko, G.\,Felder,  A.\,Gabrielov, V.\,Rubtsov, V.\,Tarasov for useful discussions. The second author
also  thanks
P.\,Etingof for recommending the first author for a post-doc position at UNC at Chapel Hill.

\section{Preliminaries}
\label{sec:prelim}
\subsection{Theta functions}
Our notations on theta functions follow \cite{FV2}. Let $z,q\in\C$ with $|q|<1$. Denote
\bea
(z;q) = \prod_{j=0}^\infty(1-zq^j).
\eea
Let $z=e^{2\pi i x}$ and $q=e^{2\pi i \tau}$. The first Jacobi theta function is defined by the formula
\bea
\theta_1(x,\tau) = i e^{\pi i (\tau/4-x)}(z;q)(q/z;q)(q;q).
\eea
It is an entire holomorphic odd function such that
\bea
\theta_1(x + n+ m \tau,\tau) = (-1)^{m+n} e^{-\pi i m^2\tau -2\pi i mu}\theta_1(x,\tau), \qquad m,n\in \Z.
\eea
It obeys the heat equation
\bea
4\pi i \frac{\der}{\der \tau }\theta_1(x,\tau)= \theta_1^{''}(x,\tau),
\eea
where \ $'=d/dx$.
Introduce
\bean
\label{1nt}
\theta(x,\tau)
&=&
 \frac{e^{-\pi i\tau/4}}{2\pi (q;q)^3}\theta_1(x,\tau) =
- \frac{e^{-\pi i x}}{2\pi i}(z;q)(q/z;q)(q;q)^{-2}
\\
\notag
&=&\frac{\sin(\pi x)}\pi (qz;q)(q/z;q)(q;q)^{-2}.
\eean
Then
\bea
\theta(x + n+ m \tau,\tau) &=& (-1)^{m+n} e^{-\pi i m^2\tau -2\pi i mu}\theta(x,\tau),
 \qquad m,n\in \Z ,
\\
\theta'(0,\tau) 
&=& 1.
\eea
Denote
\bea
\rho(x,\tau) &=& \frac{\theta'(x,\tau)}{\theta(x,\tau)},
\quad
\phantom{aaaaa}
\sigma (x, w,\tau) =
\frac{\theta(x+w,\tau)}{\theta(x,\tau)\theta(w,\tau)},
\\
\varphi(x,w,\tau) &=& \partial_x \sigma(w,-x,\tau),
\quad
\phantom{aaaa}
\eta(x) = \rho(x)^2 + \rho'(x).
\eea
In the sequel, we will often omit the $\tau$ argument.

\subsection{Collected formulas} \label{sec col form}
We have for any $m,n\in \Z$, we have
\bea
\rho(x+ n+ m \tau) &=& \rho(x) - 2\pi i m,
\\
\rho'(x+ n+ m \tau) &=& \rho'(x),
\\
\sigma(x+ n+ m \tau,w) &=& e^{-2\pi i m w} \sigma(x,w),
\\
\varphi(x+ n+ m \tau,w) &=& e^{2\pi i m w} \varphi(x,w),
\\
\varphi(x,w+ n+ m \tau) &=& e^{2\pi i m x} \big(\varphi(x,w) + 2\pi i m\, \sigma(w,-x) \big),
\\
\eta(x+ n+ m \tau,w) &=& \eta(x,w)- 4\pi i m \,\rho(x) + (2\pi i m)^2.
\eea
If $a_1,\dots,a_n$, $z_1,\dots,z_n\in \C$ and $\sum_{j=1}^n a_j = 0$,
then the function 
\bea
f(x) =\sum_{j=1}^n a_j \rho(x-z_j,\tau)
\eea
is doubly periodic:
\bea
f(x+ n+ m \tau) = f(x), \qquad m,n\in \Z.
\eea
The function $\rho'(x,\tau)$ is even, while $\theta(x,\tau)$ and $\rho(x,\tau)$ are odd.  We also have the
identities
\bea
\sigma(-x,-w) &=& -\sigma(x,w)
\\
\varphi(-x,-w) &=& \varphi(x,w).
\eea
We have the following Laurent series expansions about $x=0$,
\bea
\theta(x) &=& x + \mc O(x^3),
\\
\rho(x) &=& x^{-1} + \mc O(x),
\\
\sigma(x,w) &=& x^{-1} + \rho(w) + \mc O(x),
\\
\varphi(x,w) &=& x^{-2} + \mc O(1),
\\
\eta(x) &=& \mc O(1).
\eea

\begin{lem}
We have the identities
\bea
\varphi(x,w) &=& \sigma(w,-x)(\rho(x-w)-\rho(x)), \\
\varphi(x,0) &=& -\rho'(x).
\eea
\end{lem}

\begin{proof}
The first identity is obtained by taking the logarithmic derivative:
\[
\frac{\der_x\sigma(w,-x)}{ \sigma(w,-x)} = -\frac{\der_x\theta(w-x)}{\theta(w-x)} + 
\frac{\der_x\theta(-x)}{\theta(-x)},
\]
while the second can by seen taking Laurent expansion about $w=0$:
\[
\sigma(w,-x)(\rho(x-w)-\rho(x)) = (w^{-1} + \mc O(1))(-\rho'(x)w + \mc O(w^2))
.
\]
\end{proof}

\begin{lem} 
\label{lem id1}

Assume that $z_1,z_2\in \C$ do not differ by an element of $\Z+\tau\Z$. Then
  \bean
\label{5t}
\phantom{aaa}
  \frac{\sigma(x-z_1,w)\sigma(x-z_2,-w)}{\sigma(z_1-z_2,-w)} + \rho(x-z_2) - \rho(x-z_1) 
&=& \rho(w) - \rho(w-(z_1-z_2)),
  \\
\label{5t5}
  \sigma(x,w)\sigma(x,-w) &=& \rho'(w) - \rho'(x).
  \eean
\end{lem}
\begin{proof}
  The difference,
  \bea
  \frac{\sigma(x-z_1,w)\sigma(x-z_2,-w)}{\sigma(z_1-z_2,-w)} + \rho(x-z_2) - \rho(x-z_1) - \rho(w) + \rho(w-(z_1-z_2)),
  \eea
  is a doubly periodic function, both in the variable $x$ 
and in the variable $w$.  The difference is entire in both variables as well,
 hence it must be a constant.  Evaluating at $w=x-z_2$, this constant is seen to be zero.  
The second identity is proved similarly.
\end{proof}

\begin{rem}
  Formula \Ref{5t} is a reformulation of the following identity:
  \[
  \frac{\prod_{1\leq j<l\leq 3} \theta(x_j + x_l)}{\theta(x_1+x_2+x_3)\prod_{j=1}^3 \theta(x_j)} - \sum_{j=1}^3 \rho(x_j) + \rho(x_1+x_2+x_3) = 0,
  \]
which  can be obtained from \Ref{5t} by a change of variables.
\end{rem}

\begin{lem} 
\label{lem id2}
The function
\[
  f(x) = \rho(x-z_1)\rho(x-z_2)+ \rho(z_1-z_2)\rho(x-z_2)
  - \rho(z_1-z_2)\rho(x-z_1)
\]
satisfies the identities
\[
f(z_1) = f(z_2) = \eta(z_1 - z_2).
\]
\end{lem}

\begin{proof}
Observe  the Laurent expansion about $x=z_1$:
\bea
 && \rho(x-z_1)\left(\rho(x-z_2) - \rho(z_1-z_2)\right) 
\\
&&
\phantom{aaaaa}
=
  \left((x-z_1)^{-1} + \mc O(x-z_1)\right)\left(\rho'(z_1-z_2)(x-z_1) + \mc O((x-z_1)^2)\right),
\eea
so
\[
  f(z_1) = \rho'(z_1-z_2) + \rho(z_1-z_2)^2 = \eta(z_1-z_2).
\]
A similar argument shows $f(z_2) = \eta(z_2 - z_1) = \eta(z_1 - z_2)$.
\end{proof}

\section{The RST-operator}
\label{sec RST-o}

\subsection{The RST-operator for $\gln$}
 Consider the complex Lie algebra $\gln$ with standard basis $e_{jl}$, $j,l=1,\dots,N$.
Let $\h\subset \gln$ be the Cartan subalgebra generated by the elements $e_{jj},\, j=1,\dots,N$.
Let $\la\in \h$ with $\la = \la_1 e_{11}+\dots + \la_N e_{NN}$.  

\smallskip
Let $z=\{z_1,\dots,z_n\}\subset \C$ be numbers such that no pairwise difference belongs to $\Z+\tau\Z$. 

\smallskip

Let $V_1,\dots,V_n$ be $\gln$-modules and $V=\otimes_{k=1}^nV_k$.  For an element $g\in U(\gln)$, we write $g^{(k)} = 1 \otimes \cdots \otimes g \otimes \cdots \otimes 1$, with the element $g$ in the $k$-th factor.
Let  $V = \oplus_{\bla\in\h^*} V[\bla]$
be the the weight decomposition, where $ V[\bla] = \{v\in V\ |\ e_{jj}v = \bla (e_{jj})v \ \on{for}\ j=1,\dots,N \}$.
In particular,
\bea
V[0] = \{v\in V\ |\ e_{jj}v = 0, \,j=1\dots,N\}.
\eea
Denote $\la_{jl}=\la_j-\la_l$,
\bea
e^+_{jj}(x)
&=&
 {\sum}_{k=1}^n \rho(x-z_k) e_{jj}^{(k)},
 \\
e^+_{jl}(x)
&=&
{\sum}_{k=1}^n \si(x-z_k, \la_{jl}) e_{jl}^{(k)}, \quad \on{for}\ j\ne l.
\eea
Introduce the $N\times N$-matrix  $\L(x)=(\L_{jl}(x))$,
\bea
\L_{jj}(x) &=& e^+_{jj}(x) + {\sum} _{l\ne j} \rho(\la_{jl}) e_{ll},
\\
\L_{jl}(x)&=&
e^+_{lj}(x),
 \qquad \quad \on{for}\ j\ne l.
\eea
Here $e_{ll}$ acts on $V$ as $e_{ll}^{(1)}+ \dots + e_{ll}^{(n)}$.

The {\it universal dynamical differential operator} (or the {\it RST-operator})
is defined by the formula
\bean
\label{rdet}
\D = \cdet( \delta_{jl}\der_x - \delta_{jl}\der_{\la_j} + \L_{jl}),
\eean
where, for an $N\times N$-matrix $X=(X_{jl})$ with noncommuting entries the {\it column determinant} is defined by the formula
\bea
\cdet X = {\sum}_{\si\in S_N} (-1)^\si X_{\si(1),1} \dots  X_{\si(N),N}.
\eea
\label{RST}
Let us write
\bean
\label{udo}
\D= \der^N_x + {\sum}_{j=1}^N S_j(x) \der_x^{N-j} .
\eean
Recall the tensor product $V$ and the zero-weight subspace $V[0]\subset V$. We interpret every coefficient
 $S_j(x)$ as a function of $x$ with values  in the space of  differential operators in variables 
$\la_1,\dots,\la_N$   with coefficients in $\End(V)$.

\begin{thm}
[\cite{RST}]
Fix $\tau\in\C$ with $\on{Im} \tau>0$.
Fix $z=\{z_1,\dots,z_n\}\subset \C$ so that that no pairwise difference belongs to $\Z+\tau\Z$. 
Then for  $x \in \C$, and any $j=1,\dots,N$, the operator $S_j(x)$ well-defines 
a  differential operator in $\la_1,\dots,\la_N$  with coefficients in $\End(V[0])$.
Moreover,  for  any  $j,l\in\{1,\dots,N\}$, $u,v\in\C$, \ these differential operators $S_j(u), S_l(v)$ commute,
\bea
[S_j(u), S_l(v)]=0.
\eea
\end{thm}

\subsection{The RST-operator for $\sln$ and Bethe algebra}
In this paper, we are interested in  the $\sln$ version of the RST-operator.

\smallskip
The Lie algebra $\sl_N$ is a Lie subalgebra of $\gln$. We have  $\gln=\sln\oplus \C(e_{11}+\dots+e_{NN})$,
where $e_{11}+\dots+e_{NN}$ is a central element. 
Let $V_1,\dots,V_n$ be $\sln$-modules, thought of as $\gln$-modules,
 where the central element $e_{11}+\dots+e_{NN}$ acts by zero.
Let $V=\ox_{k=1}^nV_k$ be the tensor product of the $\sln$-modules.  In this paper {\it we consider only such tensor products.}

\smallskip
 We identify the algebra of functions 
on the Cartan subalgebra of $\sln$ with the algebra of functions in
the variables $\la_1,\dots,\la_N$, which depend only on  the differences $\la_{jl}=\la_i-\la_l$.
Indeed, the Cartan subalgebra of $\sln$ consists of elements $\la=\la_1e_{11}+\dots+\la_Ne_{NN}$
with $\la_1+\dots+\la_N=0$. Such elements are determined uniquely by the differences $\la_{jl}$ of the coordinates.

Denote by $\on{Fun}_{\sln}(V[0])$ the space of $V[0]$-valued meromorphic functions in 
the variables $\la_1,\dots,\la_N$, which depend only on  the differences $\la_{jl}=\la_i-\la_l$.

\smallskip
Each coefficient $S_j(x)$ of the RST-operator, associated with $V[0]$, defines a differential operator
acting on  $\Fun$. From now on {\it we consider  the coefficients $S_1(x)$, \dots, $S_N(x)$ as 
a family of  commuting differential operators  on  $\Fun$}, depending on the parameter $x$.

\smallskip
The commutative algebra $\B = \B^V(z_1,\dots,z_n, \tau)$ of operators on  $\Fun$  generated by the identity operator
 and the operators $\{ S_j(x)\ |\ j=1,\dots, N, \ x\in \C \}$
is called the {\it dynamical elliptic Bethe algebra} of $V$.

\subsection{The KZB operators}
Introduce the following elements of $\gln\otimes\gln$:
\bea
\Om_0 = {\sum}_{k=1}^N e_{kk}\otimes e_{kk},
\qquad
\Om_{jl} = e_{jl}\otimes e_{lj} \ {}\ \on{for}\ j\ne l,
\qquad \Om = \Om_0+{\sum}_{j\ne l}\Om_{jl}.
\eea
The KZB operators $H_0(z_1,\dots,z_n,\tau), \dots,H_n(z_1,\dots,z_n,\tau)$ as operators on the $V[0]$-valued functions on the Cartan subalgebra of $\sln$ were introduced in \cite{FW}. As differential operators on $\Fun$ they have the following form:
\[
H_0(z,\tau) = \frac{1}{4 \pi i} \sum_{k=1}^N \der^2_{\la_k} + \frac{1}{4\pi i} \sum_{s,p} \Big[\frac{1}{2}\eta(z_s - z_p,\tau) \Omega_0^{(s,p)}
- \sum_{j\ne l} \varphi(\la_{jl}, z_s - z_p, \tau) \Omega_{jl}^{(s,p)} \Big],
\]
\[
H_s(z,\tau) = -{\sum}_{k=1}^N e_{kk}^{(s)} \partial_{\la_k} + {\sum}_{p\,:\,p\neq s} \Big[ \rho(z_s - z_p,\tau) \Omega_0^{(s,p)}
+ {\sum}_{j\ne l} \sigma(z_s - z_p, -\la_{jl}, \tau) \Omega_{jl}^{(s,p)} \Big],
\]
see these formulas in Section 3.3 of \cite{JV}.

By \cite{FW} the operators $H_0(z,\tau),H_1(z,\tau), \dots, H_n(z,\tau)$ commute and
\bean
\label{eq kzb sum}
{\sum}_{s=1}^n H_s(z,\tau) = 0.
\eean

\subsection{Center of $\gln$}

Let $Z(x)$ be the following polynomial in the variable $x$ with coefficients
in $\Ugln$:
\vvn.3>
\beq
\label{Zx}
Z(x)\,=\,\rdet\left( \begin{matrix}
x-e_{11} & -\>e_{21}& \dots & -\>e_{N1}\\
-\>e_{12} &x+1-e_{22}& \dots & -\>e_{N2}\\
\dots & \dots &\dots &\dots \\
-\>e_{1N} & -\>e_{2N}& \dots & x+N-1-e_{NN}
\end{matrix}\right).
\vv.3>
\eeq
The next statement was proved in \cite{HU}, see also \cite[Section~2.11]{MNO}.

\begin{thm}
\label{Zcent}
The coefficients of the polynomial\/ $\,Z(x)-x^N$ are free generators
of the center of\/ $\Ugln$.
\qed
\end{thm}

For example, the coefficient of $x^{N-2}$ equals
\bean
\label{C2}
&&
{\sum}_{1\leq j<l\leq N} ((e_{jj}-j+1)(e_{ll}-l+1) - e_{lj}e_{jl})
\\
&&
\notag
=
{\sum}_{1\leq j<l\leq N} (e_{jj}e_{ll} - e_{lj}e_{jl})+{\sum}_{j=1}^N (j-1)e_{jj} 
- \frac{N(N-1)}2{\sum}_{j=1}^N e_{jj} + \on{const}
\\
&&
\notag
=
{\sum}_{1\leq j<l\leq N} (e_{jj}e_{ll} - e_{lj}e_{jl}+ e_{ll}) -\frac{N(N-1)}2{\sum}_{j=1}^N e_{jj} + \on{const}
\\
&&
\notag
=
{\sum}_{1\leq j<l\leq N} (e_{jj}e_{ll} - e_{jl}e_{lj} +e_{jj}) -\frac{N(N-1)}2{\sum}_{j=1}^N e_{jj} + \on{const}.
\eean

\section{First properties of the Bethe algebra}
\label{sec:first results}
\subsection{Double periodicity of the RST-operator}

\begin{thm} 
 \label{thm trans}

The operator $\D$ satisfies
\bea
\D(x+ k+ l \tau) = \D(x), \qquad k,l\in \Z,
\eea
so each $S_k(x)$ is a doubly periodic function in the variable $x$.
\end{thm}

\begin{proof} Clearly $\D(x+1)=\D(x)$. 
Let us study $\D(x+\tau)-\D(x)$.
 The diagonal part the matrix  of $\D(x+\tau)-\D(x)$ is the matrix
$-2\pi i\on{diag}(e_{11},\dots,e_{NN})$.  The off-diagonal $(jl)$-th entry of  $\D(x+\tau)$ equals the
off-diagonal $(jl)$-th entry of  $\D(x)$ multiplied by $e^{2\pi i \la_{jl}}$. 

Decompose $\D(x+\tau)$ into the sum of monomials 
\bea
M_L(x+\tau):=\D_{l_1, 1}(x+\tau)\D_{l_2,2}(x+\tau)\dots \D_{l_N,N}(x+\tau) .
\eea 
Assume that such a monomial has a factor $\D_{rr}(x+\tau)=\der_x-\der_{\la_r} + \L_{rr}(x)-2\pi i e_{rr}$ for some $r$.

\begin{lem}
The presence of the  term $-2\pi i e_{rr}$ in  $\D_{rr}(x+\tau)$ does not contribute to 
the difference $\D(x+\tau)-\D(u)$.

\end{lem}

\begin{proof}
 In this monomial the element $e_{rr}$ stays in front of
the product 
\bea
\D_{l_{r+1}, r+1}(x+\tau)\D_{l_{r+2},r+2}(x+\tau)\dots \D_{l_{N},N}(x+\tau) .
\eea 
Since $r\notin\{r+1,\dots,N, l_{r+1},\dots,l_{N}\}$,  the element $e_{rr}$ commutes with 
$\D_{l_{r+1}, r+1}(x+\tau)\times$ $\D_{l_{r+2},r+2}(x+\tau)\dots \D_{l_{N},N}(x+\tau)$.
 Hence $e_{rr}$ can 
be written as the last factor of that monomial. After that  $e_{rr}$ will act on 
the zero weight subspace by zero and the corresponding product will act by zero.
\end{proof}

\begin{lem}
The operator $\der_{\la_r}$ commutes with each of the factors $e^{2\pi i \la_{r+1,l_{r+1}}}$,\dots, $e^{2\pi i \la_{N,l_N}}$ appearing
in the product 
$\D_{l_{r+1},r+1}(x+\tau)$
\dots $\D_{l_N,N}(x+\tau) $.
\end{lem}

\begin{proof} 
The lemma follows from the fact that $r\notin\{r+1,\dots,N, l_{r+1},\dots,l_N\}$
\end{proof}

By the first of the previous lemma we see that in order to prove the theorem 
we need to check that the factors $e^{2\pi i \la_{jl}}=\D_{jl}(x+\tau)/\D_{jl}(x)$ do not change the
determinant. By the second of the  previous lemma we see that those factors behave like in a standard commutative determinant
(they are not being differentiated).
Thus we need to check that in an arbitrary monomial $M_L(x+\tau)$ the total product of the factors $e^{2\pi i \la_{jl}}$ equals zero.
But that follows from the fact that $l_1,\dots,l_N $ is a permutation of $1,\dots,N$.
\end{proof}

\subsection{Dependence upon $z$}

Recall that $z=\{z_1,\dots,z_n\}$ is an $n$-tuple of complex numbers such that no pairwise 
difference belongs to $\Z+\tau\Z$.  The RST-operator depends on $z,\tau$.  It is clear from the formulas of 
Section \ref{sec col form} that
$\D(x,z_1,\dots,z_n,\tau) = \D(x,z_1,\dots,z_s+m,\dots, z_n,\tau)$ for any  $m\in \Z$.

\begin{thm}
\label{thm onz}
We have
  \bea
  \D(x,z_1,\dots,z_s+\tau,\dots, z_n,\tau) = e^{-2 \pi i \sum_l \la_l e_{ll}^{(s)}} \D(x,z_1,\dots,z_n,\tau) e^{2 \pi i \sum_l \la_l e_{ll}^{(s)}}
  \eea
\end{thm}

\begin{proof}
Let us calculate the effect of conjugating by $A= e^{-2 \pi i \sum_l \la_l e_{ll}^{(s)}}$:
\bea
&&
A e_{jj}^{(s)} A^{-1} = e_{jj}^{(s)},
\\
&&
A e_{jl}^{(s)} A^{-1} = e^{2\pi i \la_{jl}} e_{jl}^{(s)},
\qquad j\ne l,
\\
&&
A \partial_{\la_j} A^{-1} = \partial_{\la_j} - 2 \pi i e_{jj}^{(s)}.
\eea
  Thus, we have
\bea
&&
A \L_{jj}(x,z) A^{-1} = \L_{jj}(x,z),
\\
&&
A \L_{jl}(x,z) A^{-1} = e^{2\pi i \la_{lj}} \si(x-z_s, \la_{lj}) e_{lj}^{(s)} + 
{\sum}_{k\neq s} \si(x-z_k, \la_{lj}) e_{lj}^{(k)}.
\eea
Let $\D_{jl}(x,z) = \delta_{jl}\der_x - \delta_{jl}\der_{\la_j} + \L_{jl}$.  Then
\bea
&&
A \,\D_{jj}(x,z) A^{-1} = \der_x - \der_{\la_j} + 2 \pi i e_{jj}^{(s)} + \L_{jj} = \D_{jj}(x,z) A^{-1} + 2 \pi i e_{jj}^{(s)},
\\
&&
A \,\D_{jl}(x,z) A^{-1} = e^{2\pi i \la_{lj}} \si(x-z_s, \la_{lj}) e_{lj}^{(s)} + {\sum}_{k\neq s} \si(x-z_k, \la_{lj}) e_{lj}^{(k)}
.
\eea

On the other hand, the formulas of Section \ref{sec col form} yield
\bea
&&
\L_{jj}(x,z_1,\dots,z_s+\tau,\dots, z_n) = \L_{jj}(x,z_1,\dots, z_n) + 2\pi i e_{jj}^{(s)},
\\
&&
\L_{jl}(x,z_1,\dots,z_s+\tau,\dots, z_n) =
 e^{2\pi i \la_{lj}} \si(x-z_s, \la_{lj}) e_{lj}^{(s)} + {\sum}_{k\neq s} \si(x-z_k, \la_{lj}) e_{lj}^{(k)}
,
\eea
thus we have shown that for any $j,l$, $D_{jl}(x,z_1,\dots,z_s+\tau,\dots, z_n) = A \,\D_{jl}(x,z) A^{-1}$.  
This yields the theorem.
\end{proof}

\begin{cor} \label{cor Bet alg conjug}

The two commutative algebras $\B^V(z_1,\dots,z_n, \tau)$ and 
$\B^V(z_1,\dots,z_s+\tau,\dots,$ $z_n, \tau)$ of operators on  $\Fun$  
 are conjugated by the operator
$e^{-2 \pi i {\sum}_l \la_l e_{ll}^{(s)}}$.
\end{cor}

\subsection{Laurent decomposition of the RST-operator}

\begin{thm} 
\label{lem laurent}

Let
\bea
\D(x,z,\tau) = {\sum}_{j =1}^N \frac {c_{i,j}(z,\tau)}{(x-z_i)^j} + \mc O(1), \qquad i=1,\dots,n,
\eea 
be the Laurent expansion of $\D(x,z,\tau)$  at $x=z_i$, where $c_{i,j}(z,\tau)$ are operators on $\Fun$.
  Then
\bea
\D(x,z,\tau) = {\sum}_{i=1}^n{\sum}_{j=1}^N \frac{(-1)^{j-1}}{(j-1)!}\,
\rho^{(j-1)}(x-z_i,\tau)\, c_{i,j}(z,\tau)\ + \ c_0(z,\tau),
\eea
where $\rho^{(j-1)}(x,\tau)$ is the $(j-1)$-st derivative of $\rho(x,\tau)$ with respect to $x$,
and $c_0(z,\tau)$ is an operator on $\Fun$ independent of $x$.

\end{thm}

\begin{proof}
  First, choose a fundamental parallelogram $\La$ for $\C/(\Z + \tau \Z)$ acting on $\C$
such that each $z_i$, $i=1,\dots,n$, differs by an element of $\Z + \tau \Z$ from an 
element $\tilde z_i$ in the interior of $\La$.    
Then the set 
of poles of $\D(x,z,\tau)$ in $\La$ with respect to $x$ is $\{\tilde z_1, \dots, \tilde z_n\}$. 
 Since $\D(x,z,\tau)$ is doubly periodic, the Laurent expansion of $\D(x,z,\tau)$ at $x=z_i$
 has the same coefficients as the expansion at $x=\tilde z_i$.
  
We have 
\bea
 \sum_{i=1}^n c_{i,1}(z,\tau) =\int_{\der \La} \D(x,z,\tau) dx= 0.
\eea
Now the difference $\D(x) - \sum_{i=1}^n\sum_{j=1}^N \frac{(-1)^{j-1}}{(j-1)!}\rho^{(j-1)}(x-\tilde z_i,\tau) c_{i,j}$
is an entire doubly periodic function of $x$, thus the difference is a constant with respect to $x$.
\end{proof}

\begin{cor}
The dynamical elliptic Bethe algebra
$\B^V(z,\tau)$ is generated by the $nN+1$ elements $c_{i,j}(z,\tau)$ and $c_0(z,\tau)$ of Theorem \ref{lem laurent}.
\qed
\end{cor}

\subsection{The coefficients $S_1(x)$ and $S_2(x)$}

\begin{lem}
We have
\bea
S_1(x,z,\tau) = {\sum}_{j=1}^N \L_{jj}(x,z,\tau) - {\sum}_{j=1}^N \der_{\la_j} = 0 
\eea
as an operator on $\Fun$.
\qed
\end{lem}

\begin{thm}
 \label{prop:S_2}

We have the following identity of operators on $\Fun$:
\bea
S_2(x,z,\tau) = -2\pi i H_0(z,\tau) - {\sum}_{s=1}^n \Big[H_s(z,\tau)\rho(x-z_s,\tau) + c_2^{(s)} \rho'(x-z_s,\tau) \Big],
\eea
where  $c_2 = \sum_{j<l} (e_{jj}e_{ll} - e_{jl}e_{lj} + e_{jj})$
is a central element by equation (\ref{C2}).
\end{thm}

\begin{cor} \label{cor fingen}
Assume that $N=2$ and  each $V_s$, $s=1,\dots,n$, is an irreducible $\slt$-module.  
Then the dynamical elliptic Bethe algebra $\B^V\!(z,\tau)$ is generated by the KZB operators
$H_0(z,\tau)$, \dots,$H_n(z,\tau)$.
\end{cor}

\begin{proof}
  Since all $V_s$ are irreducible, each $(c_2)^{(s)}$ acts by a scalar on $\on{Fun}_\slt(V[0])$.
\end{proof}

\begin{proof}[Proof of Theorem \ref{prop:S_2}]
Denote
\bea
F(x,\la_{jl}) &:=& S_2(x) + 2\pi i H_0(z) + {\sum}_{s=1}^n \Big[H_s(z)\rho(x-z_s) + c_2^{(s)} \rho'(x-z_s) \Big]
\\
&&
 - {\sum}_{s\neq p} \rho(z_s-z_p) \rho(x-z_s)(c_1 \otimes c_1)^{(s,p)}.
\eea
It is immediate that $F(x,\la_{jl})$ is meromorphic in $x$, having poles of at most second order at the points $z_s + \Z+\tau\Z$, $s=1,\dots,n$, and holomorphic elsewhere.  Similarly, $F(x,\la_{jl})$ is meromorphic in each varaible $\la_{jl}$, having poles of at most second order at the points $\Z+\tau\Z$, and holomorphic elsewhere.

We have
\bea
S_2(x)
&=& {\sum}_{1\leq j<l\leq N} \Big[(\L_{jj}-\der_{\la_j})(\L_{ll} - \der_{\la_l}) - \L_{lj}\L_{jl} + \L_{ll}' \Big]
\\
&=&
{\sum}_{1\leq j<l\leq N}
\big[\der_{\la_j}\der_{\la_l} - \L_{jj}\der_{\la_l}-\L_{ll}\der_{\la_j}
+\L_{jj}\L_{ll} - \L_{lj}\L_{jl} -\frac{\der \L_{ll}}{\der \la_j} + \frac{\der \L_{ll}}{\der x}
\big].
\eea
We may expand this expression
\bea
S_2(x) &=&
-\frac{1}2 \sum_j \der_{\la_j}^2 + \sum_{j<l}\Big[ \sum_s (-\rho(x-z_s) e_{jj} \der_{\la_l} - \rho(x-z_s) e_{ll} \der_{\la_j} + \rho'(x-z_s) e_{ll})^{(s)}
\\
&& + \sum_{s,p} (\rho(x-z_s)\rho(x-z_p) e_{jj} \otimes e_{ll} - \sigma(x-z_s,\la_{jl})\sigma(x-z_p,-\la_{jl}) e_{jl}\otimes e_{lj} )^{(s,p)} \Big].
\eea
The function $S_2(x)$ is doubly periodic in the variable $x$ by Theorem \ref{thm trans}.  

The coefficient of $(x-z_s)^{-2}$ in the Laurent expansion of $S_2(x)$ about $z_s$ is
$c_2^{(s)}$.
Hence $F(x,\la_{jl})$ has no second order poles in the variable $x$.

The coefficient of $(x-z_s)^{-1}$ in $S_2(x)$ is
\bea
{\sum}_{j<l} & \Big[ -(e_{jj}\der_{\la_l} + e_{ll}\der_{\la_j})^{(s)}
  + \sum_{p:p\neq s} \big[\rho(z_s-z_p)(e_{jj}\otimes e_{ll} + e_{ll}\otimes e_{jj})
  \\
  & - \sigma(z_s-z_p,-\la_{jl}) e_{jl}\otimes e_{lj}
  - \sigma(z_s-z_p,\la_{jl}) e_{lj}\otimes e_{jl} \big]^{(s,p)}
  \Big],
\eea
which equals $- H_s(z)$.
We have $\sum_s  H_s(z)= 0$  by \Ref{eq kzb sum}. Hence  the function
\bea
-{\sum}_{s=1}^n  H_s(z)  \rho(x-z_s)
\eea
is doubly periodic with respect to $x$ and has the same residues as $S_2(x)$.
Now we may conclude  that $F(x,\la_{jl})$ is an entire
doubly periodic function in $x$. Thus it is a function only depending on the variables $\la_{jl}$.

Next, we may write
\bea
F(x,\la_{jl}) 
& =&
{\sum}_{s} c_2^{(s)} \rho'(x-z_s) - {\sum}_{s\neq p} \rho(z_s-z_p) \rho(x-z_s)(c_1 \otimes c_1)^{(s,p)}
\eea
\bea
+ {\sum}_{s} \big[ 
{\sum}_{j<l} -\rho(x-z_s) (e_{jj}\der_{\la_l}+ e_{ll}\der_{\la_j}) - {\sum}_{j} \rho(x-z_s) e_{jj}\der_{\la_j} \big]^{(s)}
\eea
\bea
&& + {\sum}_{s\neq p} {\sum}_{j<l} \Big[-\si(x-z_s,\la_{jl})\si(x-z_p,-\la_{jl}) - \varphi(\la_{jl},z_s-z_p) 
\\
&& \phantom{aaaaaaaaa}+ \si(z_s-z_p,-\la_{jl})(\rho(x-z_s)-\rho(x-z_p)) \Big](e_{jl}\otimes e_{lj}) ^{(s,p)}
\\
&& + \sum_{s} \sum_{j<l} \Big[\left(-\si(x-z_s,\la_{jl})\si(x-z_s,-\la_{jl}) - \frac{1}2 \varphi(\la_{jl},0) \right) e_{jl} e_{lj} - \frac{1}2 \varphi(\la_{jl},0) e_{lj} e_{jl} \Big]^{(s)}
\eea
\bea
&& + {\sum}_{s\neq  p} {\sum}_{j<l} \rho(x-z_s)\rho(x-z_p) (e_{jj}\otimes e_{ll})^{(s,p)}
\\
&& + {\sum}_{s\neq  p} {\sum}_j \big[\frac{1}4 \eta(z_s-z_p) + \rho(x-z_s)\rho(z_s-z_p) \big] (e_{jj}\otimes e_{jj})^{(s,p)}
\eea
\bea
 + {\sum}_s {\sum}_{j<l} \big[ \rho(x-z_s)^2 e_{jj} e_{ll} + \rho'(x-z_s) e_{ll} \big]^{(s)}
 + {\sum}_s {\sum}_j \frac{1}4 \eta(0) (e_{jj}^2)^{(s)}.
\eea
The second line is
\bea
&& -{\sum}_{s} \rho(x-z_s) \big[ {\sum}_{1\leq j<l\leq N} (e_{jj}\der_{\la_l} + e_{ll}\der_{\la_j}) + {\sum}_{j} e_{jj}\der_{\la_j}\big]^{(s)}
\\
&& = -{\sum}_s \rho(x-z_s) c_1^{(s)}(\der_{\la_1}+\dots + \der_{\la_N}),
\eea
and hence is zero.

By Lemma \ref{lem id1}, the sum on the third and fourth lines is zero, while the fifth line is equal to
\bea
&& \sum_{s} \sum_{j<l} \big[\rho'(x-z_s) e_{jl}e_{lj} + \frac{1}2 \varphi(\la_{jl},0) (e_{jj} - e_{ll})\big]^{(s)}
\\
&& = \sum_{s} \sum_{j<l} \big[\rho'(x-z_s) e_{jl}e_{lj}\big]^{(s)}
+ \sum_{j<l} \frac{1}2 \varphi(\la_{jl},0) (e_{jj} - e_{ll})
 = \sum_{s} \sum_{j<l} \big[\rho'(x-z_s) e_{jl}e_{lj}\big]^{(s)},
\eea
in $D_\gl(V[0])$.  This shows that, in fact, $F(x,\la_{jl})$ does not depend on $\la_{jl}$ either.

Now, the expression for $F=F(x,\la_{jl})$ inside of $D_\sl(V[0])$ reduces to
\bea
&& \sum_{s\neq  p} \sum_{j<l} \Big[\rho(x-z_s)\rho(x-z_p)+ \rho(z_p-z_s)\rho(z_p-x)
+ \rho(z_s-z_p)\rho(z_s-x)\Big] (e_{jj}\otimes e_{ll})^{(s,p)}
\\
&& + \sum_{s\neq  p} \sum_j \frac{1}4 \eta(z_s-z_p) (e_{jj}\otimes e_{jj})^{(s,p)}
 + \sum_s \sum_{j<l} \eta(x-z_s) (e_{jj} e_{ll})^{(s)}
+ \sum_s \sum_j \frac{1}4 \eta(0) (e_{jj}^2)^{(s)}
.
\eea
To complete the proof, we must show that $F=0$.  Let us view $F$ as a $D_\gl(V[0])$-valued function in the variable $z_k$, clearly regular at $z_k = x$ and $z_k = z_j$ for each $j\neq k$.  We will show that $F$ is a doubly periodic function in each variable $z_k$.  It follows that $F$ is entire in $z_k$, hence it does not depend on the choice of these points.  Then we can degenerate to the case when all $z_k = 0$, which will show that the resulting operator acts on $V[0]$ by zero.

We write $F = F(z_k)$ to emphasize the dependence of $F$ on the variable $z_k$. We have
\bea
F(z_k+\tau) - F(z_k) =
\sum_{s\neq k} \Big[\sum_{j<l} \big[(2\pi i)^2 - 2\pi i \left(\rho(z_k-u)+\rho(z_k-z_s)\right) \big] (e_{jj}\otimes e_{ll} + e_{ll}\otimes e_{jj})
\eea
\bea
 + \frac{1}2 \sum_j \big[(2\pi i)^2 - 4\pi i \rho(z_k-z_s) \big] e_{jj}\otimes e_{jj} \Big]^{(k,s)}
 + \sum_{j<l} \big[(2\pi i)^2 - 4\pi i \rho(z_k-x) \big] (e_{jj}e_{ll})^{(k)}
\eea
\bea
= \sum_{s\neq k} \Big[\sum_{j\neq l} \big[(2\pi i)^2 - 2\pi i \left(\rho(z_k-x)\right) \big] (e_{jj}\otimes e_{ll})
+ \frac{1}2 \sum_j \big[(2\pi i)^2 \big] e_{jj}\otimes e_{jj} \Big]^{(k,s)}
\eea
\bea
&&
+ \frac{1}2 \sum_{j\neq l} \big[(2\pi i)^2 - 4\pi i \rho(z_k-x) \big] (e_{jj}e_{ll})^{(k)}
\\
&&
\phantom{aaaaaaaa}
-2\pi i \sum_{s\neq k} \rho(z_k-z_s)\Big[\sum_{j\neq l} (e_{jj}\otimes e_{ll}
+ \sum_j e_{jj}\otimes e_{jj} \Big]^{(k,s)}
\eea
\bea
= \Big[\sum_{j\neq l} \big[(2\pi i)^2 - 2\pi i \left(\rho(z_k-x)\right) \big] (e_{jj}^{(k)} e_{ll} - (e_{jj} e_{ll})^{(k)})
+ \frac{1}2 \sum_j (2\pi i)^2  (e_{jj}^{(k)} e_{jj} - (e_{jj}^2)^{(k)}) \Big]
\eea
\bea
 + \frac{1}2 \sum_{j\neq l} \big[(2\pi i)^2 - 4\pi i \rho(z_k-x) \big] (e_{jj}e_{ll})^{(k)}
-2\pi i \sum_{s\neq k} \rho(z_k-z_s)\big[c_1 \otimes c_1 \big]^{(k,s)}
\eea
\bea
&&
=  2\pi^2 {(c_1)^2}^{(k)}
+ \sum_{j\neq l} \big[(2\pi i)^2 - 2\pi i \left(\rho(z_k-x)\right) \big] e_{jj}^{(k)} e_{ll}
\\
&&
\phantom{aaa}
+ \frac{1}2 \sum_j (2\pi i)^2 e_{jj}^{(k)} e_{jj}
 -2\pi i \sum_{s\neq k} \rho(z_k-z_s)\big[c_1 \otimes c_1 \big]^{(k,s)},
\eea
hence it is zero in $D_\gl(V[0])$.

Now, evaluating at $x=z_1=\dots = z_n=0$ using the Lemma \ref{lem id2}, we obtain the following expression for $F$.
\bea
F &=& \sum_{s\neq  p} \sum_{j<l} \eta(0) (e_{jj}\otimes e_{ll})^{(s,p)}
 + \sum_{s\neq  p} \sum_j \frac{1}4 \eta(0) (e_{jj}\otimes e_{jj})^{(s,p)}
\\
&& + \sum_s \sum_{j<l} \eta(0) (e_{jj} e_{ll})^{(s)}
+ \sum_s \sum_j \frac{1}4 \eta(0) (e_{jj}^2)^{(s)}
= \frac{1}4 \eta(0) \big[ (c_1)^2 + 2 \sum_{j<l} e_{jj}e_{ll}],
\eea
which is zero in $D_\gl(V[0])$.
\end{proof}

\subsection{Weyl group invariance}
The Weyl group $W$ acts on the Cartan subalgebra of $\sln$ and on the space $V[0]$
in the standard way.  Hence the Weyl group acts on $\Fun$  by the formula
\bea
s : \psi(\la) \mapsto s.(\psi(s^{-1} . \la)),
\eea
for $s\in W$, $\psi \in\Fun$.
This extends to a Weyl group action on the space 
\\
 $\End(\Fun)$, where for  
$T \in \End(\Fun)$ and $s\in W$, the operator $s(T)$ is defined as the product 
$sTs^{-1}$ of the three elements of $\End(\Fun)$.

\begin{lem} 
\label{weyl inv}
For the Lie algebra $\slt$, the RST-operator $\D(x,z,\tau)\in \End(\on{Fun}_{\slt}(V[0]))$ is Weyl group invariant.
\end{lem}

\begin{proof}
By \cite{FW} the KZB operators  $H_0(z,\tau),\dots,H_n(z,\tau)$ are Weyl group invariant.
Now the lemma  follows from Theorem \ref{prop:S_2}.
\end{proof}

\begin{cor}
All elements of the $\slt$ dynamical elliptic Bethe algebra  $\B^V(z,\tau)$ are  Weyl group invariant.

\end{cor}

The Weyl group invariance of the RST-operator for $\sln$ will be discussed elsewhere.

\section{The RST-operator and Bethe ansatz for $\slt$}
\label{sec:D}

\subsection{Representations}
We consider the zero weight subspace $V[0]$ of the tensor product of $\slt$ representations,
 where 
$V=\ox_{s=1}^n V_{m_s}$ and $V_{m_s}$ is the finite-dimensional
irreducible representation with highest weight $m_s$, considered as an $\frak{gl}_2$-module on which
the central element $e_{11}+e_{22}$ acts by zero.
The dimension of $V[0]$ is positive if the sum $\sum_{s=1}^n m_s$ is even. We denote this sum by  $2m$,
\bean
\label{sum 2m}
{\sum}_{s=1}^n m_s = 2m.
\eean

\subsection{Bethe ansatz}
Let 
\bean
\label{xi mu}
\xi = \frac{\mu}{2} \al,
\eean
 where $\mu\in \C$ and 
$\al$ is the simple root of $\slt$,  $\langle \al,e_{11}\rangle=1$, $\langle\al,e_{22}\rangle=-1$.

The {\it elliptic master function} (see Section 5 of \cite{FV1}) associated to 
$\mu\in\C$, $z=(z_1,\dots,z_n)\in\C^n$ is the
 following function of $\mu$,  $t=(t_1,\dots,t_m)$,  $z$, $\tau$:
\bean
\label{master}
&&
\Phi(\mu,t,z,\tau)= \frac{\pi i}2 \mu^2 \tau  +
 2\pi i \mu \,\Big(\sum_{i=1}^mt_i -\sum_{s=1}^n \frac{m_s}2  z_s\Big) 
+\ 2\sum_{1\leq i<j\leq m} \ln \theta(t_i-t_j,\tau)
\\
\notag
&&
\phantom{aaa}
 - \sum_{i=1}^m \sum_{s=1}^nm_s\ln\theta(t_i-z_s,\tau)
+\sum_{1\leq s<r\leq n}\frac{m_sm_r}2\ln \theta(z_s-z_r,\tau).
\eean
The {\it Bethe ansatz equations} are the equations for the
critical points of the master function
$\Phi(\mu,t,z,\tau)$ with respect to the variables $t$,
\bean
\label{BAE}
2\pi i \mu + 2 \sum_{j\ne i}\rho(t_i-t_j,\tau)
-\sum_{s=1}^n m_s\rho(t_i-z_s,\tau)=0,
\qquad
i=1,\dots,m.
\eean

\begin{thm} [\cite{FV1}]
\label{FV thm}

Let $(\mu^0, t_1^0,\dots,t_m^0,z_1^0,\dots,z_n^0,\tau^0)$
 be a solution of the Bethe ansatz equations \Ref{BAE}. Then
there is a $V[0]$-valued meromorphic function of $\la_{12}$, denoted by 
$\Psi(\la_{12},\mu^0, t^0,  z^0,\tau^0)$, such that
\bea
H_a(z^0,\tau^0)\, \Psi(\la_{12},\mu^0, t^0,  z^0,\tau^0) 
&=&
 \frac {\der \Phi}{\der z_a}(\mu^0, t^0,  z^0,\tau^0) \,
 \Psi(\la_{12}, \mu^0, t^0,   z^0,\tau^0), \qquad a=1,\dots,n,
\\
H_0(z^0,\tau^0) \Psi(\la_{12},\mu^0, t^0,  z^0,\tau^0) 
&=&
 \frac {\der \Phi}{\der \tau}(\mu^0, t^0,  z^0,\tau^0)\, \Psi(\la_{12},
 \mu^0, t^0,   z^0,\tau^0).
\eea
Thus, $\Psi(\la_{12},\mu^0, t^0,   z^0,\tau^0)$ is a meromorphic eigenfunction of the KZB operators with the eigenvalues
given by the 
partial derivatives of the master function with respect to the parameters $z$, $\tau$.

\end{thm}

\begin{proof}
In \cite{FV1}, integral representations for solutions of the KZB equations 
\bea
\kappa \partial_{z_a} \psi 
&=& H_a(z,\tau) \psi, \qquad a=1,\dots,n,
\\
\kappa \partial_{\tau} \psi &=& H_0(z,\tau) \psi,
\eea
are constructed starting from horizontal sections of a suitable local system,
see \cite[Proposition 5]{FV1}. 
An example of a horizontal section is given by the function
\bean
\label{hors}
e^{\pi i \mu \la_{12}}\, 
\Phi(\mu, t, z,\tau).
\eean
As explained in \cite{RV}, integral representations of solutions of the Knizhnik
Zamolodchikov type
equations can be used to construct common eigenvectors of the commuting
systems of operators staying in  the right-hand sides of the equations, by applying
the stationary phase method to the integral, when $\ka\to 0$, see \cite[Section 7]{FV1}.
Applying this procedure to the horizontal section
in \Ref{hors}, one obtains the eigenfunction $\Psi(\la_{12},\mu, t, z,\tau)$ of Theorem \ref{FV thm},
see in \cite{FV1} a formula for $\Psi(\la_{12},\mu, t, z,\tau)$.  A formula for  $\Psi(\la_{12},\mu, t, z,\tau)$
in the special case
$V=\ox_{s=1}^n V_{1}$ of the tensor product of two-dimensional irreducible $\slt$-modules  is given in
 Section \ref{sec:eigenfunctions}.
\end{proof}

\subsection{Fundamental differential operator}
\label{sec Fdo}

Let  $\Psi(\la_{12})$  be  a $V[0]$-valued eigenfunction of $S_2(x)$,
\bean
\label{EIg}
S_2(x) \Psi(\la_{12}) = B_2(x) \Psi(\la_{12}),
\eean
where $B_2(x)$ is a  scalar function of $x$.
We assign to $\Psi$ a scalar  differential operator with respect to the variable $x$,
\bean
\label{fund Psi}
\D_{\Psi} =\der_x^2 + B_2(x),
\eean
called the {\it  fundamental differential operator} of the eigenfunction $\Psi$.

\begin{lem}
\label{lem fund} The fundamental differential operator $\D_\Psi$ is doubly periodic,
\bean
\label{dpf}
\D_{\Psi}(x+\tau)=\D_{\Psi}(x+1)=\D_{\Psi}(x).
\eean
Moreover, if $s(\D_{\Psi})$ is the image of $\D_{\Psi}$ under the Weyl involution, then
$\D_{s(\Psi)}=\D_{\Psi}$.
\end{lem}

\begin{proof}
The lemma is a corollary of Theorem \ref{thm trans} and Lemma \ref{weyl inv}.
\end{proof}

In particular, let  $(\mu, t, z,\tau)$ be a solution of the Bethe ansatz equations \Ref{BAE}
(we will omit the index $^0$). Let $\Psi(\la_{12},\mu,t, z)$ be the corresponding eigenfunction
of  the KZB operators, see Theorem \ref{FV thm}. By Theorem \ref{prop:S_2}
the coefficient  $S_2(x)$ is a linear combination of the KZB operators. Hence 
$\Psi(\la_{12},\mu,t, z)$ is an eigenfunction of  $S_2$. Let 
\bean
\label{fund t}
\D_{(\mu,t,z,\tau)} =\der_x^2
+ B_2(x, \mu, t, z ,\tau)
\eean
be the fundamental differential operator of  $\Psi(\la_{12},\mu,t, z)$.
This operator will be also called the {\it fundamental differential operator} of the solution
 $(\mu, t, z,\tau)$.

\begin{rem}
In the non-dynamical setting the fundamental differential operator of a solution
 of the Bethe ansatz equations was introduced in
\cite{ScV, MV1}.

\end{rem}

We will give a formula for  $\D_{(\mu,t,z,\tau)}$.   Let 
\bean
\label{def u}
y(x) ={\prod}_{i=1}^m\theta(x-t_i,\tau),
\qquad
u(x)= e^{\pi i \mu x}y(x){\prod}_{s=1}^n\theta(x-z_s,\tau)^{-m_s/2}.
\eean

\begin{thm}
\label{prop fund diff op}
We have 
\bean
\label{B-R}
B_2(x, \mu, t,  z ,\tau)
= - (\ln u)\rq{}\rq{} - ((\ln u)\rq{})^2,
\eean
where the function $u(x)$ is defined in \Ref{def u}.  In other words,
\bean
\label{D_t}
\D_{(\mu, t, z,\tau)} =\big(\der_x + (\ln u)\rq{}\big)\big(\der_x - (\ln u)\rq{}\big).
\eean

\end{thm}

\begin{proof}
First, we calculate $B_2(x, \mu, t, z,\tau)$.
\begin{lem}
We have
\bea
4\pi i \frac{\der}{\der\tau}(\ln \theta(t-z,\tau)) =  \eta(t-z,\tau) - \eta(0).
\eea
\end{lem}
\begin{proof}
Recall that we have defined
\bea
\theta(u,\tau) = \theta_1(u,\tau) / \theta_1\rq{}(0,\tau)
,
\eea
where the first Jacobi theta function $\theta_1(u,\tau)$ obeys the heat equation
\bea
4\pi i \frac{\der}{\der \tau }\theta_1(u,\tau)= \theta_1^{''}(u,\tau).
\eea
Since $\theta_1\rq{}(u,\tau)$ is holomorphic in $u$ and $\tau$, we have
\bea
4\pi i \frac{\der}{\der \tau}\theta_1\rq{}(u,\tau)
= 4\pi i \frac{\der}{\der u} \frac{\der}{\der \tau}\theta_1(u,\tau)
= \theta_1^{'''}(u,\tau).
\eea

Thus
\bea
4\pi i \frac{\der}{\der\tau}(\ln \theta(t-z,\tau))
 = \frac{\theta_1^{''}(t-z,\tau)}{\theta_1(t-z,\tau)}
- \frac{\theta_1^{'''}(0,\tau)}{\theta_1\rq{}(0,\tau)}
 = \eta(t-z,\tau) - \eta(0).
\eea
\end{proof}

The eigenvalue of the operator  $-4 \pi i H_0(z,\tau)$
on the eigenfunction $\Psi(\la_{12},t, \mu,  z,\tau)$ equals 
\bea
&&
-2(\pi i)^2 \mu^2 - 4\pi i \frac{\der}{\der\tau}
\Big[2\sum_{i<j} \ln \theta(t_i-t_j) - \sum_{i,s} m_s \ln \theta(t_i-z_s)
+ \sum_{s<r} \frac{m_s m_r}2 \ln \theta(z_s-z_r)\Big]
\\
&&
= -2(\pi i)^2 \mu^2
- 2 \sum_{i<j} (\eta(t_i-t_j) - \eta(0))
+ \sum_{i,s} m_s (\eta(t_i-z_s) - \eta(0))
\eea
\bea
 - \sum_{s<r} \frac{m_s m_r}2 (\eta(z_s-z_r)-\eta(0))
= -2(\pi i)^2 \mu^2
- 2 \sum_{i<j} \eta(t_i-t_j)
+ \sum_{i,s} m_s \eta(t_i-z_s)
\\
&&
- \sum_{s<r} \frac{m_s m_r}2 \eta(z_s-z_r)
+ \big(m(m-1) - m \sum_s m_s + \frac{1}2 \sum_{s<r} m_s m_r\big) \eta(0)
\eea
\bea
&&
= -2(\pi i)^2 \mu^2
- 2 \sum_{i<j} \eta(t_i-t_j)
+ \sum_{i,s} m_s \eta(t_i-z_s)
- \sum_{s<r} \frac{m_s m_r}2 \eta(z_s-z_r)
\\
&&
\phantom{aaaaaa} - \big(m(m + 1) - \frac{1}2 \sum_{s<r} m_s m_r\big) \eta(0)
.
\eea
For $s=1,\dots,n$, the operator $H_s(z,\tau)$ has eigenvalue
\bea
\frac{\der\Phi}{\der z_s} = -\pi i m_s \mu - \sum_{i=1}^m m_s \rho(z_s-t_i)
+ \sum_{r\ne s} \frac{m_s m_r}2 \rho(z_s-z_r)
,
\eea
so we have
\bea
&&
B_2(x, \mu, t,  z ,\tau) =  - (\pi i)^2 \mu^2 - \sum_{i<j} \eta(t_i-t_j)
+ \sum_{i,s} \frac{m_s}2 \eta(t_i-z_s) - \sum_{s<r} \frac{m_s m_r}4 \eta(z_s-z_r)
\\
&&
- \frac{1}4 \Big(2 m(m + 1) - \sum_{s<r} m_s m_r\Big) \eta(0)
+ \pi i \mu \sum_{s=1}^n m_s \rho(x-z_s)
+ \sum_{s,i} m_s \rho(z_s - t_i) \rho(x-z_s)
\\
&&
\phantom{aaaaa}
- \sum_{s\neq r} \frac{m_s m_r}2 \rho(z_s-z_r) \rho(x-z_s)
+ \frac{1}4 \sum_s m_s(m_s+2) \rho'(x-z_s).
\eea
Next we calculate the right-hand side in \Ref{B-R}, namely,  the function 
$R(x, \mu, t,  z ,\tau)=- (\ln u)\rq{}\rq{} - ((\ln u)\rq{})^2$. We have
\bea
(\ln u)\rq{}
&=&  
\pi i \mu + {\sum}_{i=1}^m \rho(x-t_i) - \frac{1}2 {\sum}_{s=1}^m m_s \rho(x-z_s),
\\
-(\ln u)\rq{}\rq{} 
&=&
 - {\sum}_i \rho\rq{}(x-t_i) + \frac{1}2 {\sum}_s m_s \rho\rq{}(x-z_s),
\eea
\bea
&&
-((\ln u)\rq{})^2 = -(\pi i)^2 \mu^2 - {\sum}_i \rho(x-t_i)^2 - \frac{1}4 {\sum}_s m_s^2 \rho(x-z_s)^2
\\
&&
\phantom{aaa}
+ \pi i \mu \Big({\sum}_s m_s \rho(x-z_s)  - 2 {\sum}_i \rho(x-t_i)\Big)
- {\sum}_{i\ne j} \rho(x-t_i)\rho(x-t_j)
\\
&&
\phantom{aaa}
- \frac{1}4 {\sum}_{s\ne r} m_s m_r \rho(x-z_s)\rho(x-z_r)
+{\sum}_{i,s} m_s \rho(x-t_i)\rho(x-z_s),
\eea
\bea
&& 
R(x,  \mu, t,  z ,\tau)=
-\frac d{dx}\ln\rq{}u -(\ln\rq{}u)^2 =  -(\pi i)^2 \mu^2 - \sum_i \eta(x-t_i)
 - \frac{1}4 \sum_s m_s^2 \eta(x-z_s)
\\
&&
\phantom{aaa} + \frac{1}4 {\sum}_s m_s(m_s + 2) \rho\rq{}(x-z_s)
+ \pi i \mu \Big({\sum}_s m_s \rho(x-z_s)  - 2 {\sum}_i \rho(x-t_i)\Big)
\\
&&
\phantom{aaa} - {{\sum}}_{i\ne j} \rho(x-t_i)\rho(x-t_j)
- \frac{1}4 {{\sum}}_{s\ne r} m_s m_r \rho(x-z_s)\rho(x-z_r)
\\
&&
\phantom{aaaaaaaaaa} +{\sum}_{i,s} m_s \rho(x-t_i)\rho(x-z_s)
.
\eea
The function $R(x, t, \mu, z ,\tau)$ is doubly periodic in $x$, since
\bea
&&
R(x+\tau, t, \mu, z ,\tau) - R(x, t, \mu, z ,\tau)
\\
&&
\phantom{aa} =  {\sum}_i \big[4\pi i \rho(x-t_i) - (2\pi i)^2\big]
+ {\sum}_s  m_s^2 \big[\pi i \rho(x-z_s) - (\pi i)^2\big]
\\
&&
\phantom{aaaa} - \pi i \mu \Big({\sum}_s 2\pi i m_s  - 2 {\sum}_i 2\pi i\Big)
+ {\sum}_{i\ne j} \big[2\pi i \big(\rho(x-t_i)+\rho(x-t_j)\big) - (2\pi i)^2\big]
\eea
\bea
&&
+ {\sum}_{s\ne r} m_s m_r \big[\frac{\pi i}2 \big(\rho(x-z_s)+\rho(x-z_r)\big) - (\pi i)^2\big]
\\
&&
\phantom{aaaa}
+ {\sum}_{i,s} m_s \big[ (2\pi i)^2 - 2\pi i\big(\rho(x-t_i) + \rho(x-z_s)\big)\big]
\eea
\bea
&&
= (\pi i)^2 \big[-4m - {\sum}_s m_s^2 - \mu(2{\sum}_s m_s - 4m)- 4m(m-1)
- {\sum}_{s\ne r} m_s m_r + 8m^2 \big]
\\
&&
\phantom{aaa}
+ \pi i {\sum}_i \big(4 + 4(m-1)- 2{\sum}_s m_s\big) \rho(x-t_i)
\\
&&
\phantom{aaaaaa}
+ \pi i {\sum}_s \big(m_s^2 + m_s {\sum}_{r\ne s} m_r - 2m m_s\big) \rho(x-z_s) = 0.
\eea
Consider the Laurent expansion of $R(x)$ at each pole.  Clearly the coefficient in $R(x)$ of $(x-t_i)^{-2}$ is zero.
The coefficient in $R(x)$ of $(x-t_i)^{-1}$ equals 
\bea
-2\pi i \mu - 2 {{\sum}}_{j\ne i} \rho(t_i-t_j) + {\sum}_s m_s \rho(t_i-z_s)
= 0
.
\eea
Hence $R(x)$ is regular at $x=t_i$ for all $i$.

The coefficient in $R(x)$ of $(x-z_s)^{-2}$ equals  $- \frac{1}4 m_s(m_s + 2)$.
The coefficient in $R(x)$ of $(x-z_i)^{-1}$ equals
\bea
\pi i \mu m_s - \frac{1}2 {\sum}_{r\ne s} m_s m_r \rho(z_s-z_r)
+{\sum}_i m_s \rho(z_s-t_i).
\eea
These calculations show that the function $B_2(x, \mu, t, z ,\tau)$ 
has the same set of poles in $x$ and the same Laurent tails  at each pole  in $x$ as
the function $R(x, \mu, t,  z ,\tau)$. 
Since both functions are doubly periodic in $x$, we may conclude that
$B_2(x, \mu, t,  z ,\tau) - R(x, \mu, t,  z ,\tau)$ 
is constant in $x$.  We need to show that this difference is zero.

We will use the following notion of the constant term of 
a meromorphic doubly periodic function $F(x)$, regular in the complement to the union of
 the $\Z+\tau\Z$-orbits  of the points 
$z_1,\dots,z_n$. Namely, let
\bea
F(x) = {\sum}_{j =1}^\infty \frac {c_{i,j}}{(x-z_i)^j} + \mc O(1), \qquad i=1,\dots,n,
\eea 
be the Laurent expansion of $F(x)$ at $x=z_i$. Then
\bea
F(x) = {\sum}_{i=1}^n{\sum}_{j=1}^N  \frac{(-1)^{j-1}}{(j-1)!}\,\rho^{(j-1)}(x-z_i,\tau) c_{i,j}\ +\ c_0,
\eea
where the number  $c_0\in \C$ will be  called the {\it constant term} of $F(x)$,  cf. Theorem \ref{lem laurent}.

We need to show that the constant term of $B_2(x, \mu, t,  z ,\tau)$ :
\bea
&&
- (\pi i)^2 \mu^2 - \sum_{i<j} \eta(t_i-t_j)
+ \sum_{i,s} \frac{m_s}2 \eta(t_i-z_s)
- \sum_{s<r} \frac{m_s m_r}4 \eta(z_s-z_r)
\\
&&
\phantom{aaaaa} - \frac{1}4 \Big(2 m(m + 1) - \sum_{s<r} m_s m_r\Big) \eta(0)
\eea
equals the constant term of $R_2(x, \mu, t,  z ,\tau)$:
\bea
&&
\notag  -(\ln u)\rq{}\rq{} -((\ln u)\rq{})^2
 - {\sum}_s \frac{m_s(m_s+2)}4 \rho'(x-z_s)
\\
&&
\notag
\phantom{aaaaa} - {\sum}_s \Big[\pi i \mu m_s
- {\sum}_{r\ne s} \frac{m_s m_r}2 \rho(z_s-z_r)
+{\sum}_i m_s \rho(z_s-t_i)\Big] \rho(x-z_s)
\eea
\bea
&&
\notag
= -(\pi i)^2 \mu^2 - {\sum}_i \eta(x-t_i) - \frac{1}4 {\sum}_s m_s^2 \eta(x-z_s)
\\
&&
\notag
\phantom{aaaaa} + \frac{1}4 {\sum}_s m_s(m_s + 2) \rho\rq{}(x-z_s)
- 2 \pi i \mu {\sum}_i \rho(x-t_i)
\eea
\bea
&&
\notag
+ {\sum}_s \Big[\pi i \mu m_s
- \frac{1}4 {\sum}_{r\ne s} m_s m_r \rho(x-z_r)
+ {\sum}_i m_s \rho(x-t_i)
\Big] \rho(x-z_s)
\\
&&
\notag
\phantom{aaaaa}
- {\sum}_{i\ne j} \rho(x-t_i)\rho(x-t_j)
- {\sum}_s \frac{m_s(m_s+2)}4 \rho\rq{}(x-z_s)
\eea
\bea
&&
\notag
 - {\sum}_s \Big[\pi i \mu m_s
- {\sum}_{r\ne s} \frac{m_s m_r}2 \rho(z_s-z_r)
+{\sum}_i m_s \rho(z_s-t_i)\Big] \rho(x-z_s)
\\
&&
\notag
\phantom{aaaaa}
= -(\pi i)^2 \mu^2 - {\sum}_i \eta(x-t_i) - \frac{1}4 {\sum}_s m_s^2 \eta(x-z_s)
\eea
\bea
&&
\notag
 - 2 \pi i \mu {\sum}_i \rho(x-t_i)
- {\sum}_{i\ne j} \rho(x-t_i)\rho(x-t_j)
\\
&&
\notag
\phantom{aaaaa} 
+ {\sum}_{s\ne r} \frac{m_s m_r}4 \big(2\rho(z_s-z_r) - \rho(x-z_r)\big)\rho(x-z_s)
\eea
\bea
&&
\notag
\phantom{aaaaa} +{\sum}_{i,s} m_s \big(\rho(x-t_i) - \rho(z_s-t_i)\big)\rho(x-z_s)
\eea
\bean
\label{constcalc}
&&
\phantom{aaa}
=\ -(\pi i)^2 \mu^2 - {\sum}_i \eta(x-t_i) - \frac{1}4 {\sum}_s m_s^2 \eta(x-z_s)
\\
&&
\notag
\phantom{aaaaaa} + {\sum}_{s,i} m_s \big[\rho(t_i-x) \rho(t_i-z_s) + \rho(x-t_i) \rho(x-z_s) + \rho(z_s - x) \rho(z_s-t_i) \big]
\\
&&
\notag
\phantom{aaaaaa} -2 {\sum}_{i < j} \big[\rho(t_i-t_j)\rho(t_i-x) + \rho(x-t_j) \rho(x-t_i)
+ \rho(t_j-t_i) \rho(t_j-x)\big]
\\
&&
\notag
\phantom{aaaaaa} - {\sum}_{s<r} \frac{m_s m_r}2 \big[ \rho(z_s-z_r)\rho(z_s-x)
+ \rho(x-z_r) \rho(x-z_s) + \rho(z_r-z_s) \rho(z_r-x) \big] 
.
\eean
In the calculation of the expression in \Ref{constcalc} we use the Bethe ansatz equation to substitute
\bea
-2 \pi i \mu {\sum}_i \rho(x-t_i) = \Big[ 2 {\sum}_{j\ne i} \rho(t_i-t_j) - {\sum}_s m_s \rho(t_i-z_s)\Big] {\sum}_i \rho(x-t_i)
.
\eea
According to the above calculations the constant term of $B_2(x, t, \mu, z ,\tau)\! -\! R(x, t, \mu, z ,\tau)$
equals the following expression:
\bea
&&
f(t, z) = -{\sum}_{i<j} \eta(t_i-t_j)
+ {\sum}_{i,s} \frac{m_s}2 \eta(t_i-z_s)
- {\sum}_{s<r} \frac{m_s m_r}4 \eta(z_s-z_r)
\\
&&
\phantom{aaa} - \frac{1}4 \Big(2 m(m + 1) + {\sum}_{s<r} m_s m_r\Big) \eta(0)
+ {\sum}_i \eta(x-t_i) + \frac{1}4 {\sum}_s m_s^2 \eta(x-z_s)
\eea
\bea
&&
- {\sum}_{s,i} m_s \big[\rho(t_i-x) \rho(t_i-z_s) + \rho(x-t_i) \rho(x-z_s) + \rho(z_s - x) \rho(z_s-t_i) \big]
\\
&&
\phantom{aaa} +2 {\sum}_{i < j} \big[\rho(t_i-t_j)\rho(t_i-x) + \rho(x-t_j) \rho(x-t_i)
+ \rho(t_j-t_i) \rho(t_j-x)\big]
\eea
\bea
\phantom{aaa} + {\sum}_{s<r} \frac{m_s m_r}2 \big[ \rho(z_s-z_r)\rho(z_s-x)
+ \rho(x-z_r) \rho(x-z_s) + \rho(z_r-z_s) \rho(z_r-x) \big] 
.
\eea

\begin{lem}
\label{lem f reg}
The function $f(t, z)$ is regular and  doubly periodic in each of the variables 
$t_1$, \dots, $t_m$, $z_1,\dots,z_n$ and hence is a constant in $t$ and $z$.

\end{lem}

\begin{proof} The regularity is easily checked by calculating the residues. The periodicity is checked as follows:
\bea
f(\dots,t_i+\tau,\dots) - f(\dots,t_i,\dots) =
{\sum}_{j\ne i} \big[ 4\pi i \rho(t_i - t_j) - (2\pi i)^2 \big]
\eea
\bea
&&
\phantom{aaaaa}
+ {\sum}_s \frac{m_s}2 \big[(2\pi i)^2 - 4\pi i \rho(t_i-z_s) \big]
+ \big[(2\pi i)^2 - 4\pi i \rho(t_i-x)\big]
\\
&&
\phantom{aaaaa}
+{\sum}_s m_s \big[2 \pi i \big(\rho(t_i-x) + \rho(t_i-z_s)\big) - (2\pi i)^2\big]
\\
&&
\phantom{aaaaa}
+2 {\sum}_{j\ne i} \big[(2\pi i)^2 - 2\pi i\big(\rho(t_i-t_j)+\rho(t_i-x)\big)\big]
\eea
\bea
&& = (\pi i)^2 \big[ -4 (m-1) + 2 {\sum}_s m_s + 4 - 4{\sum}_s m_s + 8(m-1) \big]
\eea
\bea
&&
\phantom{aaaaa}
+ 4 \pi i {\sum}_{j\ne i} \big(1 - 1 \big) \rho(t_i-t_j)
+ \pi i {\sum}_s \big(-2 m_s + 2m_s \big) \rho(t_i-z_s)
\\
&&
\phantom{aaaaa}
+ \pi i \big(-4 + 2 {\sum}_s m_s - 4(m-1) \big)\rho(t_i-x) =0;
\eea
\bea
&&
f(\dots,z_s+\tau,\dots) - f(\dots,z_s,\dots) =
{\sum}_i \frac{m_s}2 \big[(2\pi i)^2 - 4\pi i\rho(z_s-t_i)\big]
\\
&&
\phantom{aaaaa}
+ {\sum}_{r\ne s} \frac{m_s m_r}4 \big[4\pi i \rho(z_s-z_r) - (2\pi i)^2 \big]
+ \frac{1}4 m_s^2 \big[(2\pi i)^2 - 4\pi i\rho(z_s-x)\big]
\\
&&
\phantom{aaaaa}
+ {\sum}_i m_s \big[2\pi i\big(\rho(z_s-x) + \rho(z_s-t_i)\big) - (2\pi i)^2\big]
\\
&&
\phantom{aaaaa}
+\frac{1}2 {\sum}_{r\ne s} m_s m_r \big[(2\pi i)^2 - 2\pi i\big(\rho(z_s-z_r)+\rho(z_s-x)\big)
\big]
\\
&& = (\pi i)^2 \big[2 m m_s - m_s {\sum}_{r\ne s} m_r + m_s^2 - 4m m_s + 2 m_s {\sum}_{r\ne s} m_r \big]
\\
&&
\phantom{aaaaa}
+ \pi i {\sum}_i \big(-2m_s +2m_s \big) \rho(z_s-t_i)
+ \pi i {\sum}_{r\ne s} \big(m_s m_r - m_s m_r\big) \rho(z_s-z_r)
\\
&&
\phantom{aaaaa}
+ \pi i \big(-m_s^2 + 2m m_s - {\sum}_{r\ne s}m_s m_r \big)\rho(z_s-x) = 0.
\eea
\end{proof}

By Lemma \ref{lem f reg} the function $f(t, z )$ is 
 constant in  $t$ and $z$, while its value  at $t=0$, $z=0$ is
\bea
&&
\eta(0)\Big[-\frac{m(m-1)}2
+ m^2
- {\sum}_{s<r} \frac{m_s m_r}4
- \frac{m(m + 1)}2
\\
&&
\phantom{aaaaa}  + \frac{1}4 {\sum}_{s<r} m_s m_r
+ m + \frac{1}4 {\sum}_s m_s^2 - 2m^2 +m(m-1) +\frac{1}2 {\sum}_{s<r} m_s m_r
\Big]=0,
\eea
see Lemma \ref{lem id2}.
Hence $f(t, z)$ equals zero and Theorem \ref{prop fund diff op} is proved.
\end{proof}

\section{Special case $ V=(\C^2)^{\ox 2m}$}
\label{sec 6}

In the remainder of this paper  we consider the situation of Section \ref{sec:D}.  
In addition to the assumptions of Section \ref{sec:D} we will always assume  that 
\bea
m_1 = \dots = m_ n = 1, \quad \on{ and \,so} \quad n= 2m;
\eea
  that is, from now on  we will study the dynamical elliptic 
  Bethe algebra $\B^V(z_1,\dots,z_{2m}, \tau)$ on the zero-weight subspace of the tensor product 
  \bea
  V=(\C^2)^{\ox 2m}
  \eea
   of   two-dimensional irreducible $\slt$-modules.

In this case the  elliptic master function \Ref{master} becomes
\bean
\label{mast s}
&&
\Phi(t,\mu,z,\tau)= \frac{\pi i}2 \mu^2 \tau  +
 2\pi i \mu \,\Big({{\sum}}_{i=1}^mt_i -{\sum}_{s=1}^{2m} \frac{1}2  z_s\Big) 
 +\ 2\sum_{1\leq i<j\leq m} \ln \theta(t_i-t_j,\tau) 
\\
\notag
&&
\phantom{aaa}
- {\sum}_{i=1}^m {\sum}_{s=1}^{2m}\ln\theta(t_i-z_s,\tau)
+\sum_{1\leq s<r\leq 2m}\frac{1}2\ln \theta(z_s-z_r,\tau),
\eean
the Bethe ansatz equations \Ref{BAE} 
become
\bean
\label{BAEs}
2\pi i \mu +2{\sum}_{k,\, k\ne j}\rho(t_j-t_k,\tau) - {{\sum}}_{s=1}^{2m} \rho(t_j-z_s,\tau) = 0,  \qquad j=1,\dots,m.
\eean

\section{Theta-polynomials}
\label{sec:wron}

\subsection{Definitions}

Fix $\tau\in \C$ with $\on{Im}\tau>0$.

\begin{defn}
\label{def1}

A {\it theta-polynomial of degree $m$} is a function of the form
\bean
\label{thp}
f(x) = c e^{2\pi i \mu x}{\prod}_{j=1}^m\theta(x-t_j,\tau)
\eean
where $c, \mu, t_1,\dots,t_m \in\C$.
We have
\bean
\label{trp}
&&
f(x+1) = e^{2\pi i \mu}(-1)^m f(x),
\\
\notag
&&
f(x+\tau) = e^{2\pi i\mu\tau}(-1)^m e^{-\pi i m\tau -2\pi i m x  + 2\pi i \sum_{j=1}^mt_j} f(x).
\eean
The {\it first and second multipliers} of the theta-polynomial are the numbers
\bean
\label{mltp}
A=e^{2\pi i \mu },
\qquad   B=e^{2\pi i \mu \tau + 2\pi i\sum_{j=1}^mt_j}.
\eean
\end{defn}

\begin{defn}
\label{def2}
Given  $A, B\in\C^\times$, an entire function  $f(x)$ is called a
 {\it theta-polynomial of degree }
$m$ with multipliers $A,\,B$
if
\bean
\label{n-tranf}
f(x+1) = A(-1)^mf(x),
\qquad
f(x+\tau) = B (-1)^m e^{-\pi i m \tau -2\pi i m x } f(x).
\eean

\end{defn}

Clearly if $f(x)$ satisfies Definition \ref{def1} then it satisfies Definition \ref{def2}.

\begin{lem}
\label{lem determin}
Let $f(x)$ satisfies Definition \ref{def2} with multipliers $A,B$. Let $\mu\in\C$ be such that
$A=e^{2\pi i \mu}$. Then there exist $t_1,\dots,t_m\in\C$ such that
\bean
\label{d2}
f(x) = ce^{2\pi i \mu x}{\prod}_{j=1}^m\theta(x-t_j,\tau).
\eean

\end{lem}

\begin{proof} The function  $g(x)=e^{-2\pi i \mu x}f(x)$ has the transformation properties:
\bea
g(x+1) = (-1)^m g(x), \qquad
g(x+\tau) = e^{-2\pi i \mu \tau} B (-1)^m e^{-\pi i m \tau -2\pi i m x } g(x).
\eea
Let $s_1,\dots,s_m\in\C$ be any numbers such that 
\bea
e^{2\pi i \sum_{j=1}^ms_j} =  e^{-2\pi i \mu \tau}  B.
\eea 
Then the product $h(x)=\prod_{j=1}^m\theta (x-s_j,\tau)$ has the transformation properties:
\bea
&&
h(x+1)=(-1)^mh(x),
\\
&&
h(x+\tau) =  
 (-1)^m
 e^{-\pi i m \tau -2\pi i m x +2\pi i \sum_{j=1}^ms_j}h(x)=
e^{-2\pi i \mu \tau} B (-1)^m
 e^{-\pi i m \tau -2\pi i m x }h(x).
 \eea
 The function $g/h$ is doubly periodic  with $m$ poles at the points $[s_1], \dots, [s_m]$ 
on
the elliptic curve $\C/(\Z + \tau \Z)$ and 
$m$ zeros  at some points $[t_1],\dots,[t_m]$ on the elliptic curve $\C/(\Z + \tau \Z)$. 
By Theorem 20.14 in \cite{WW} we have $\sum_{j=1}^m [t_j] = \sum_{j=1}^m [s_j]$ on 
 $\C/(\Z + \tau \Z)$.
Choose representatives  $t_1,\dots,t_m\in \C$ of  $[t_1],\dots,[t_m]$  such that 
 $\sum_{j=1}^m s_j = \sum_{j=1}^m t_j$.
Define $u(x) = \prod_{j=1}^m\theta(x-t_j,\tau)$. Then
\bea
&&
u(x+1)=(-1)^mu(x),
\\
&&
u(x+\tau) =  
 (-1)^m
 e^{-\pi i m \tau -2\pi i m x +2\pi i \sum_{j=1}^mt_j}u(x)=
 (-1)^m  e^{-\pi i m \tau -2\pi i m x +2\pi i \sum_{j=1}^ms_j}u(x).
\eea
Hence
\bea
\frac{g(x)}{h(x)} \frac {h(x)}{u(x)} =\frac{g(x)}{u(x)}
\eea
is an entire doubly periodic function. Hence it is some constant $c\in\C$ and we obtain \Ref{d2}.
\end{proof}

Notice that the numbers $c, \mu, t_1,\dots,t_m$ are not unique.

\begin{lem}
 \label{lem 1.2}
  If
\bean
\label{2p}
ce^{2\pi i \mu x} {\prod}_{j=1}^m\theta(x-t_j,\tau)
\quad\on{and}\quad
c'e^{2\pi i \mu' x} {\prod}_{j=1}^m\theta(x-t_j',\tau)
\eean
are presentations of the same function, then 
\begin{enumerate}
\item[(i)]
 after a suitable permutation of $t_1,\dots,t_m$ we have
$t_j'=t_j-k_j-l_j\tau$ for $j=1,\dots,m$ and some $k_j,l_j\in\Z$;

\item[(ii)]  $\mu=\mu'- \sum_{j=1}^m l_j$;

\item[(iii)]  
\bean
\label{cc}
c= c'(-1)^{\sum_{j=1}^m(k_j+l_j)}e^{-\pi i \sum_{j=1}^mk_j^2\tau + 2\pi i \sum_{j=1}^m k_jt_j} .
\eean

\end{enumerate}

Conversely, for any $c,c'$, $\mu, \mu'$, $t_1, \dots, t_m$, $t_1',\dots,t_m'$  satisfying conditions (i-iii)
the two functions in \Ref{2p} are equal.
\end{lem}

\begin{proof} Statement (i) is obvious. The second presentation in \Ref{2p} takes the form
\bea
c'e^{2\pi i \mu'x} (-1)^{\sum_{j=1}^m(k_j+l_j)}
e^{-\pi i \sum_{j=1}^mk_j^2\tau + 2\pi i \sum_{j=1}^m k_jt_j} e^{-2\pi i \sum_{j=1}^m k_jx}
 {\prod}_{j=1}^m\theta(x-t_j,\tau).
 \eea
Comparing the two presentations we obtain the lemma.
\end{proof}

\subsection{Space $T_{m,A,B}$}

Given $m\in\Z_{>0}$ and $A,B\in\C^\times$, denote by $T_{m,A,B}$ 
the vector space of theta-polynomials of degree $m$ with multipliers $A,B$.

\begin{lem}
\label{lem dim}
The vector space $T_{m,A,B}$ has dimension $m$.

\end{lem}

\begin{proof} Let $\nu\in\C$ be such that $e^{2\pi i\nu} = A(-1)^m$. Then
for any $f\in T_{n,A,B}$, the function $g(x)=e^{-2\pi i \nu x}f(x)$ has the properties:
\bea
g(x+1) = g(x),
\qquad
g(x+\tau) = e^{-2\pi i\nu\tau} B (-1)^m e^{-\pi i m \tau -2\pi i m x } g(x).
\eea
Let $g(x) = \sum_{k\in \Z} a_k  e^{2\pi i k x}$ be
the Fourier expansion and $q=e^{2\pi i \tau}$. Then
\bea
{\sum}_{k\in \Z} a_k q^k e^{2\pi i k x}
=
e^{-2\pi i \nu\tau} B (-1)^m e^{-2\pi i m x} q^{-m/2} {\sum}_{k\in \Z} a_k e^{2\pi i k x}
\eea
or
\bean
\label{coe even}
a_{k+m} = e^{2\pi i\nu\tau} B^{-1}(-1)^m q^{k+m/2} a_k.
\eean
Arbitrary choice of coefficients  $a_1,\dots, a_m$ determines $g(x)$ uniquely.
The map sending a theta-polynomial to the vector $(a_1,\dots,a_m)$ identifies 
$T_{m,A,B}$ with $\C^m$.
\end{proof}

 For  $\nu\in\C$ the map
\bean
\label{Liso}
L_\nu : T_{m, A,B} \to T_{m, e^{2\pi i \nu}A, \,e^{2\pi i \nu\tau}B},
\qquad f(x)\mapsto e^{2\pi i \nu x}f(x)
\eean
is an isomorphism of vector spaces.

\subsection{Spaces  $T_{m,A}$,  $T_{m}$}

For $A\in\C^\times$ denote by $T_{m,A} $ the  space of theta-polynomials 
of degree $m$ with first  multiplier $A$ and by $T_{m} $ the  space of all theta-polynomials 
of degree $m$. 
Fix $\mu\in\C$ such that $e^{2\pi i \mu}=A$. Then all elements of $T_{m,A}$ have the form
\bea
f(x) = e^{2\pi i \mu x}{\prod}_{j=1}^m\theta(x-t_j,\tau)
\eea
with arbitrary $t_1,\dots,t_m\in\C$, see Lemma \ref{lem determin}.
 We have
\bea
T_{m,A} =\cup_{B\in \C^\times} T_{m,A,B},
\quad
T_m=\cup_{A,B\in\C^\times}T_{m,A,B} = \cup_{A\in\C^\times} T_{m,A}.
\eea
The maps $L_\nu$, $\nu\in\C$, define an action of $\C$ on $T_m$.

\subsection{Projectivization}
\label{sec pro}

Denote by $P_{m,A,B}$ the space of nonzero theta-polynomials of degree $m$ with multipliers
$A, B$ identified up to multiplication by a nonzero number. Thus $P_{m,A,B}$ is  $m-1$-dimensional
projective space.

Denote by $P_{m,A}$ the space of nonzero theta-polynomials of degree $m$ with first multiplier
$A$ identified up to multiplication by a nonzero number and by 
 $P_{m}$ the space of nonzero theta-polynomials of degree $n$ 
 identified up to multiplication by a nonzero number.  We have
\bea
P_{m}=\cup_{A\in\C^\times} P_{m,A} = \cup_{A,B\in\C^\times} P_{m,A,B}.
\eea
The maps $L_\nu$, $\nu\in\C$, define an action of the group $\C$ on $P_m$. In particular
the maps $L_k$, $k\in\Z$, define an action of $\Z$ on each $P_{m,A}$.

The group $\C$ acts on $P_{m}\times P_m$ by the maps
$L_\nu\times L_\nu$. This action preserves  
\bea
\cup_{A_1,A_2\in\C^\times,\, A_1\ne A_2} P_{m,A_1} \times P_{m,A_2} .
\eea
One of the spaces that we are interested in this paper is the space
\bean
\label{PPC}
(\cup_{A_1,A_2\in\C^\times,\, A_1\ne A_2} P_{m,A_1} \times P_{m,A_2} )/\C
 \cong (P_{m,1}\times (P_m-P_{m,1}))/\Z
\eean
of dimension $2m+1$.

\subsection{Analytic involution}
\label{sec ani}

Define the {\it analytic involution}
\bean
\label{An In}
P_m\times P_m \to P_m\times P_m, \qquad (f,g) \mapsto (g,f).
\eean
The analytic involution commute with the diagonal action of $\C$ on 
$P_m\times P_m$ by the maps $\La_\nu\times L_\nu$, $\nu\in\C$.

The analytic involution on $P_m\times P_m$ descends to the {\it analytic
involution}
\bean
\label{an inv}
\iota_{\on{an}} : (P_{m,1}\times (P_m-P_{m,1}))/\Z  \to (P_{m,1}\times (P_m-P_{m,1}))/\Z 
\eean
as follows. Take a point 
\bea
p\in  (P_{m,1}\times P_{m,A^{-1}})/\Z \subset (P_{m,1}\times (P_m-P_{m,1}))/\Z
\eea
with  $A\ne 1$.
 Choose its representative
$(f(x),g(x))\in P_{m,1}\times P_{m,A^{-1}}$. Let $\mu\in\C$ be such that
$e^{2\pi i \mu}=A$. Define $\iota_{\on{an}}(p)$ to be the equivalence class of 
$(e^{2\pi i \mu x}g(x), e^{2\pi i \mu x}f(x))$ in $(P_{m,1}\times P_{m,A}))/\Z$. 
We have $\iota_{\on{an}}^2=\on{Id}$.

\section{Wronskian determinant}
\label{sec Wrd}

\subsection{Definition}
For functions $f(x)$, $g(x)$ the determinant
\bea
\Wr( f,g) = f(x)g\rq{}(x) - f\rq{}(x)g(x)
\eea
is called the {\it Wronskian determinant}.
We have 
\bean
\label{h^2}
\Wr( hf, hg)=h^2\Wr( f,g)
\eean
for any function $h$.

\begin{lem}
If $f\in T_{m_1,A_1,B_1}$, $g\in T_{m_2,A_2,B_2}$, then $fg \in T_{m_1+m_2, A_1A_2,B_1B_2}$.
\qed
\end{lem}

\begin{lem}
\label{lem 1.5}
If $f\in T_{m,A_1,B_1}$, $g\in T_{m, A_2,B_2}$,   then
\bea
\Wr(f, g) = h,
\eea
where $h\in T_{2m,A_1A_2,B_1B_2}$.

\end{lem}

\begin{proof} The function
\bea
\Wr( f,g)/fg= g'/g- f'/f 
\eea
is doubly periodic. Hence the entire function $\Wr(f,g)$ has the same transformation properties as
$fg\in T_{2m,A_1A_2,B_1B_2}$.
\end{proof}

\subsection{Wronskian equation}

\begin{lem} \label{lem:wr}
If $f(x),g(x)$ are holomorphic functions,
then the meromorphic function
$(g/f)'=\Wr(f,g)/f^2$ has zero residue at every pole.
\qed
\end{lem}

\begin{thm}
[\cite{BMV}]
\label{BMV}  Let $f\in T_{m,A_1,B_1}$,  $h \in T_{2m, A_1A_2,B_1B_2}$ be 
nonzero functions such that $(A_1,B_1)\ne (A_2,B_2)$. 
Let all zeros of $f(x)$ be simple. Let the function
$h/f^2$ have zero residue at every zero of $f$.
 Then there exists a unique $g\in T_{m,A_2,B_2}$ 
such that $\Wr(f, g) = h$.
\end{thm}

\begin{proof} The function $m= h/f^2$ satisfies equations
\bea
m(x+1) = A\, m(x),
\qquad
m(x+\tau) = B\, m(x),
\eea
where $A= A_2/A_1, B=B_2/B_1$ and hence $(A,B)\ne (1,1)$. Choose $x_0\in\C$ not a pole of $m$ and define
$ M(x) = \int_{x_0}^x m(u) du$.
Then
\bea
M(x+1) = A \,M(x) + a,
\qquad
a=\int_{x_0}^{x_0+1}m(u) du ,
\\
M(x+\tau) = B\, M(x) + b,
\qquad
b = \int_{x_0}^{x_0+\tau} m(u) du .
\eea
\begin{lem}
We have
\bean
\label{reson}
a(B-1) = b(A-1).
\eean
\end{lem}

\begin{proof}
The integral of $m$ over the boundary of the parallelogram with vertices at $x_0$, $x_0+1, x_0+1+\tau,
x_0+\tau$ is zero since
$m$ has no residues. On the other hand it equals
\bea
&&
\int_{x_0}^{x_0+1} m(x)dx +\int_{x_0+1}^{x_0+1+\tau} m(x)dx + 
\int_{x_0+1+\tau}^{x_0+\tau} m(x)dx +\int_{x_0+\tau}^{x_0} m(x)dx
\\
&&
=\int_{x_0}^{x_0+1} m(x)dx +A\int_{x_0}^{x_0+\tau} m(x)dx + 
B\int_{x_0+1}^{x_0} m(x)dx +\int_{x_0+\tau}^{x_0} m(x)dx
\\
&&
=a + Ab -Ba - b.
\eea
\end{proof}

\begin{lem}
Let $M(x)$ be a function such that
\bea
M(x+1) = A \,M(x) + a,
\qquad
M(x+\tau) = B\, M(x) + b,
\eea
for some $A,B\in\C^\times$, $(A,B) \ne (1,1)$, and $a,b\in\C$. Then there exists a unique $C\in\C$ such that
the function $\tilde M(x):=M(x) + C$ satisfies the equations
\bea
\tilde M(x+1) = A \,\tilde M(x),
\qquad
\tilde M(x+\tau) = B\, \tilde M(x).
\eea
\end{lem}
\begin{proof}
For any $C$ we have
\bea
\tilde M(x+1) = A \,\tilde M(x) + C(1-A)+a,
\qquad
\tilde M(x+\tau) = B\, \tilde M(x) + C(1-B) + b.
\eea
We want $C$ to satisfy the system of equations
\bea
C(1-A)+a = 0,
\qquad
C(1-B) + b=0.
\eea
If $(A,B)\ne (1,1)$,  this  system has a solution if and only if equation \Ref{reson} holds.
\end{proof}

The function $g(x)=f(x)\tilde M(x)$ is holomorphic
since all zeros of $f$ are simple. The function $g$ 
 lies in $T_{n,A_2,B_2}$ and satisfies the equation $\Wr(f,g)=h$.
\end{proof}

\subsection{Points with generic coordinates}

We say that a point $(f,g)\in P_{m}\times P_m$ 
has {\it generic first} (resp. {\it second}) {\it coordinate } if $f$ 
(resp. $g$) has no common zeros with $\Wr(f,g)$.

\begin{lem} Assume that  $(f,g)\in P_m\times P_m$ has generic first coordinate, then
\begin{enumerate}
\item[(i)]
every point of the $\C$-orbit of $(f,g)$ has generic first coordinate;
\item[(ii)]
 all zeros of $f$ are simple;
\item[(iii)]
$f$  has no common zeros with $g$.

\end{enumerate}
\qed
\end{lem}

\begin{lem} 
\label{lem:gc}

The complement in $P_{m,1}\times (P_m-P_{m,1})$ to the
subset  of points with generic both first
and second  coordinates
is a proper analytic subset of $P_{m,1}\times (P_m-P_{m,1})$.

\end{lem}

\begin{proof} Clearly, the complement is an analytic subset of $P_{m,1}\times (P_m-P_{m,1})$.
Let us show that it is proper. Indeed, every point $f$ of $P_{m,1}$ has a presentation 
$f(x)= \prod_{j=1}^m\theta(x-t_j,\tau)$ for some $t_1,\dots,t_m\in\C$ by 
Lemma \ref{lem determin}.
Every point $g$ of $P_m-P_{m,1}$ lies in some
$P_{m,A}$ with $A\ne 1$. If $\mu\in\C$ is such that $e^{2\pi i \mu}=A$, then
the every point $g$
of $P_{m,A}$ has a presentation of the form $g= e^{2\pi i \mu x}\prod_{j=1}^m\theta(x-s_j,\tau)$
for some $s_1,\dots,s_m\in\C$ by 
Lemma \ref{lem determin}.
If $t_1,\dots,t_m, s_1,\dots,s_m $ are all distinct modulo $\Z+\tau\Z$, then the
point $(f,g) \in P_{m,1}\times (P_m-P_{m,1})$ has generic both first and second coordinates.
\end{proof}

Denote
\bea
\on{Pairs}^1_m &=& \{ p\in (P_{m,1}\times (P_m-P_{m,1}))/\Z \ | \
p \ \on{has\ generic\ first\ coordinate}\},
\\
\on{Pairs}^2_m &=& \{ p\in (P_{m,1}\times (P_m-P_{m,1}))/\Z \ | \
p \ \on{has\ generic\ second\ coordinate}\},
\\
\on{Pairs}_m  &=& \{ p\in (P_{m,1}\times (P_m-P_{m,1}))/\Z \ | \
p \ \on{has\ generic\ first\ and \ second\ coordinates}\}.
\eea

\subsection{Functions $g/f$, $\Wr(f,g)/f^2$}
Consider the meromorphic functions of the form
\bean
\label{FF}
F = g/f,
\qquad
G = \Wr(f,g)/f^2,
\eean
where $(f, g)\in P_{m,1}\times (P_m-P_{m,1})$. We
{\it will consider such functions 
 up to multiplication by nonzero numbers}.
 Notice that $G = F'$.

\smallskip

\begin{lem}
\label{lem pts}  The pairs $(f,g) \in \on{Pairs}^1_m$ 
are in bijective correspondence  with the  functions  $G = \Wr(f,g)/f^2$, 
which have $m$ distinct poles of order 2 in a fundamental parallelogram
for the lattice $\Z+\tau\Z$ acrting on $\C$.
\end{lem}

\begin{proof} Let $(f,g)\in P_{m,1}\times (P_m-P_{m,1})$ have generic first coordinate. A pair
 in the same 
$\Z$-orbit is  $(e^{2\pi i k x}f,\,e^{2\pi i k x}g)$ for some $k\in\Z$. Both pairs define the same function
$G$ by formula \Ref{h^2}.

Conversely assume that there are two  pairs
$(f,g)$, $(\hat f,\hat g)\in P_{m,1}\times (P_m-P_{m,1})$
with generic first coordinates  such that
 \bean
 \label{W=W}
 {\Wr(f,g)/f^2}=\Wr(\hat f, \hat g)/\hat f^2
 \eean
up to a constant factor.
 By assumption, the functions $f$ and $\hat f$ have the same zeros. Hence
 $\hat f =e^{2\pi i k x} f$ up to a constant factor. Then
 $\Wr(\hat f,\hat g)=e^{4\pi i k x} \Wr(f,g)$ 
up to a constant factor. 
 Then  
 $\Wr(e^{2\pi i k x}f,\hat g)=e^{4\pi i k x}\Wr(f,g)$
 and 
 $\Wr(f, e^{-2\pi i kx}\hat g)=\Wr(f,g)$. By Theorem \ref{BMV} we have
 $ e^{-2\pi i k x}\hat g = g$ up to a constant factor. The lemma is proved.
\end{proof}

\begin{lem}
\label{lem ra}  The pairs $(f,g)\in \on{Pairs}^1_m$ 
are in bijective correspondence  with functions  $F = g/f$,
 which have $m$ simple poles  in a fundamental parallelogram for the lattice $\Z+\tau\Z$ acting on $\C$.
\end{lem}

\begin{proof} Let $(f,g)\in P_{m,1}\times (P_m-P_{m,1})$ have generic first coordinate. 
A pair in the same 
$\Z$-orbit is  $(e^{2\pi i k x}f,\,e^{2\pi i k x}g)$ for some $k\in\Z$. Both pairs define the same function
$F$.

Conversely assume that there are two  points
$(f,g), (\hat f,\hat g)\in P_{m,1}\times (P_m-P_{m,1})$
with generic first coordinates  such that the corresponding two ratios are equal,
 $ g/f=\hat g/\hat f.$
  Then their derivatives are also equal, see formula \Ref{W=W}. Hence the two points lie in the same $\Z$-orbit by 
 Lemma \ref{lem pts}. The lemma is proved.
\end{proof}

\subsection{Fundamental differential operator}
\label{sec fp}

For $(f,g) \in P_{m,1}\times (P_m-P_{m,1})$ introduce
two functions
\bean
\label{u12}
u_1(x) =f/\sqrt{\Wr(f,g)},
\qquad
 u_2(x) = g/\sqrt{\Wr(f,g)}.
\eean
Then 
\bean
\label{Wr1}
\Wr(u_1,u_2) = 1,
\eean
see \Ref{h^2}.
Let $\D_{(f,g)}$ be the monic linear  differential operator, whose kernel is generated by the functions
$u_1, u_2$. This operator will be called the {\it fundamental differential operator} 
of the pair $(f,g)$.

\begin{lem}
\label{lem:fdo} 
We have
\bean
\label{fdo}
\D_{(f,g)} = \Big(\der_x + (\ln u_1)\rq{}\Big)\Big(\der_x - (\ln u_1)\rq{}\Big)=
\Big(\der_x + (\ln u_2)\rq{}\Big)\Big(\der_x - (\ln u_2)\rq{}\Big).
\eean
\end{lem}

\begin{proof}
The functions $u_1,u_2$ are linearly independent. Hence
 $\D_{(f,g)}$ is a second order differential operator. 
We will prove that $\D_{(f,g)}$ equals
\bean
\label{fdou}
\Big(\der_x + (\ln u_1)\rq{}\Big)\Big(\der_x - (\ln u_1)\rq{}\Big)
=
\der_x^2 - (\ln u_1)'' - ((\ln u_1)\rq{})^2. 
\eean
The remaining equality in \Ref{fdo} is proved similarly.

Clearly the function $u_1$ lies in the kernel of the operator in \Ref{fdou}.  Since the differential operator 
$\der_x$ enters the right-hand side in \Ref{fdou} with zero coefficient, we conclude that the kernel of the operator in \Ref{fdou}
is generated by two functions,
  one of which is $u_1$ and the other, say $u$ is such that $\Wr(u_1,u)=1$, but $u_2$ is such a function.
\end{proof}

We have the following properties of the  fundamental
differential  operator  $\D_{(f,g)}$.

\begin{lem}
\label{lem pfu} ${}$

\begin{enumerate}
\item[(i)]
 The operator  $\D_{(f,g)}$ is doubly periodic,
\bea
\D_{(f,g)}(x+k+l\tau)=\D_{(f,g)}(x), \qquad k,l\in\Z
\eea
All singular points of  $\D_{(f,g)}$ are regular singular.

\item[(ii)] The monodromy group of the kernel of the operator $\D_{(f,g)}$ preserves
the direct decomposition $\langle u_1\rangle\oplus \langle u_2\rangle$. 
If $f\in P_{m, A_1,B_1}$, where $A_1=1$, and
$g\in P_{m, A_2,B_2}$, then the monodromy operator of 
the transformation
$x\to x+1$ has the matrix  $\on{diag}(\sqrt{A_1/A_2}$, 
$\sqrt{A_2/A_1})$
and the monodromy operator of the transformation
$x\to x+\tau$ has the matrix  $\on{diag}(\sqrt{B_1/B_2}, \sqrt{B_2/B_1})$.

\item[(iii)]  Assume that $(f,g)$ has generic first and second coordinates and all zeros of $\Wr(f,g)$ are simple.
Then the singular points of $\D_{(f,g)}$ are the $\Z+\tau\Z$-orbits of the set
$\{z_1,\dots,z_{2m}\}$. The local monodromy around every singular point
has the matrix $\on{diag}(-1,-1)$.

\end{enumerate}
\end{lem}

\begin{proof}
The operator $\D_{(f,g)}$ is doubly periodic  because the functions $u_1,u_2$ are doubly periodic. The other statements are clear.
\end{proof}

\begin{lem}
\label{lem foZ} The fundamental differential operator is well-defined for pairs in $\on{Pairs}^1_m$.

\end{lem}

\begin{proof} Clearly the pairs $(f,g)$ and $(e^{2\pi i k x}f, e^{2\pi i k x}g)$ have the same functions
$u_1, u_2$, see formula \Ref{h^2}. 
\end{proof}

Recall the analytic involution 
\bea
\iota_{\on{an}}: (P_{m,1}\times (P_m-P_{m,1}))/\Z  \to (P_{m,1}\times (P_m-P_{m,1}))/\Z 
\eea
defined in \Ref{an inv}.

\begin{thm}
\label{lem fodp} For any $p \in \on{Pairs}_m$ the points 
$p$ and $\iota_{\on{an}}( p)$ have the same
 fundamental differential operators. Conversely, if 
 two points $p, \tilde p\in \on{Pairs}_m$   have the same
 fundamental differential operator, then either $\tilde p=p$ or
$\iota_{\on{an}}(\tilde p)=p$.

\end{thm}

\begin{proof}  The first statement is clear from the definition of the fundamental
 differential operator. Let us prove the second statement.

For $p$ and $\tilde p$ consider  the corresponding functions
$u_1,u_2$ and $\tilde u_1, \tilde u_2$, which generate
the kernels of the corresponding 
fundamental differential operators. 

The transformation $x \to x+1$ defines a linear map on the space $\langle u_1,u_2\rangle$
with eigenvectors $u_1,u_2$, which have distinct eigenvalues. Hence
$u_1$ is proportional to $\tilde u_1$ and $u_2$ is proportional to $\tilde u_2$ or
$u_1$ is proportional to $\tilde u_2$ and $u_2$ is proportional to $\tilde u_1$.

In the first case, the function $u_1^2$ is proportional to $\tilde u_1^2$ and then $\tilde p= p$ by Lemma
\ref{lem pts}. 
In the second case,  the function  $u_1^2$ is proportional to $\tilde u_2^2$ and
$\iota_{\on{an}}(\tilde p)=p$ by Lemma \ref{lem pts}. 
\end{proof}

\section{Bethe ansatz equations and pairs of theta-polynomials} 

\label{sec:eigenfunctions}

\subsection{Bethe ansatz equations and Wronskian equation}

First notice that if
$(\mu, t_1$, \dots, $t_m, z_1,\dots,z_{2m})$ is a solution of the Bethe ansatz equations
\Ref{BAEs}, then $t_1,\dots,t_m$ are distinct  modulo the lattice $\Z+\tau\Z$ and 
any $t_j$ is distinct from any $z_a$ modulo the lattice $\Z+\tau\Z$.

\begin{lem}
 \label{lem:res}
 
 Given $\mu,$  $t_1,\dots,t_m, z_1,\dots,z_{2m}$ let $t_1,\dots,t_m$
 be distinct numbers modulo the lattice $\Z+\tau\Z$ and let
any $t_j$ be distinct from any $z_a$ modulo the lattice $\Z+\tau\Z$. Then the function 
\bean
\label{fF}
G(x) = 
e^{-2\pi i \mu x}\,
\frac{\prod_{a=1}^{2m}\theta (x-z_a,\tau)} {\prod_{j=1}^{m}\theta (x-t_j,\tau)^2}
\eean
 has zero residues with respect to $x$ at any point,  
 if and only if the numbers  $\mu, t_1,\dots,t_m$, $z_1,\dots,z_{2m}$ satisfy the Bethe ansatz equations
\Ref{BAEs}.
\end{lem}

\begin{proof} We have
\bea
G(x+1) = e^{-2\pi i \mu}\,G(x),\qquad
G(x+\tau) = e^{-2\pi i \mu\tau}e^{2\pi i(\sum_{s=1}^{2m}z_s-2 \sum_{j=1}^mt_j) }\,G(x).
\eea
Hence it is enough to check that $G(x)$ has zero residues at the points
$x=t_j, j=1,\dots,m$.  
We have $\theta (x) = \theta'(0) x  +  o(x^2)$ as $x\to 0$, since
$\theta(x)$ is odd. As $x-t_j \to 0$, we have
\bea
G(x,\mu,t,z,\tau) = e^{-2\pi i \mu x } \frac{\prod_{a=1}^{2m}\theta (x-z_a)} 
{(\theta'(0) (x-t_j)  + o((x-t_j)^2)^2\prod_{k\ne j }\theta (x-t_k)^2} .
\eea
Hence $\Res_{x=t_j}G(x)=0$ if and only if the logarithmic derivative of
\bea
e^{-2\pi i \mu x}\,\frac{\prod_{a=1}^{2m}\theta (x-z_a)} 
{\prod_{k\ne j }\theta (x-t_k)^2}
\eea
 at $x=t_j$ equals zero. That is exactly the $j$-th Bethe ansatz equation.
\end{proof}

\begin{lem}
\label{p-B}
Let  $(f,g) \in P_{m,1}\times(P_m-P_{m,1})$ have generic first coordinate and
$f=\prod_{j=1}^m\theta(x-t_j,\tau)$, $\Wr(f,g) =
e^{-2\pi i \mu x } \prod_{a=1}^{2m}\theta (x-z_a,\tau)$ for some
$\mu, t_1,\dots,t_m, z_1,\dots,z_{2m}$. Then
$(\mu, t_1$, \dots, $t_m,z_1,\dots,z_{2m})$ is a solution of the Bethe ansatz equations \Ref{BAEs}.

\end{lem}

\begin{proof} The function $(g/f)'= \Wr(f,g)/f^2$ has zero residues and satisfies the 
assumptions of Lemma \ref{lem:res}. This implies Lemma
\ref{p-B}.
\end{proof}

\begin{lem}
\label{lem 1.11}
If $(\mu, t_1,\dots, t_m,z_1,\dots,z_{2m})$
 is a solution of the Bethe ansatz equations and $\mu\not\in\Z$,
then there exist $s_1,\dots,s_m$ such that 
\bean
\label{wrr}
\phantom{aaa}
\Wr \Big(\prod_{j=1}^m\theta(x-t_j,\tau), e^{-2\pi i \mu x}\prod_{j=1}^m\theta(x-s_j,\tau)\Big)
 = \on{const} e^{-2\pi i \mu x}\prod_{a=1}^{2m}\theta(x-z_a,\tau).
\eean
Moreover,  the function $\prod_{j=1}^m\theta(x-s_j,\tau)$ is
unique up to multiplication  by a constant.
\end{lem}

\begin{proof} The first multipliers of the functions $\prod_{j=1}^m\theta(x-t_j)$  and
 $e^{-2\pi i \mu x}\prod_{a=1}^{2m}\theta(x-z_a)$ are distinct. 
 Hence there exists a unique theta-polynomial $g$ satisfying the equation 
 \bean
\label{wr}
\Wr \Big({\prod}_{j=1}^m\theta(x-t_j), g(x)\Big)
 = e^{-2\pi i \mu x}{\prod}_{a=1}^{2m}\theta(x-z_a),
\eean
see  Lemma \ref{lem:res} and Theorem \ref{BMV}.
 The theta-polynomial $g$ has degree $m$ and the first multiplier
$A=e^{- 2\pi i \mu }$, by Theorem \ref{BMV}. Such a  function  $g$ has a 
presentation
\bean
\label{nft}
g(x) = \on{const} e^{-2\pi i \mu x} {\prod}_{k=1}^m\theta(x-s_k)
\eean
for suitable $s_1,\dots,s_m\in\C$, by Lemma \ref{lem determin}.
\end{proof}

\begin{cor} 
\label{cor 9.4}

Under the assumptions of Lemma \ref{lem 1.11} we also have
\bean
\label{wrrr}
\Wr \Big(e^{2\pi i \mu x} \prod_{j=1}^m\theta(x-t_j,\tau), \prod_{j=1}^m\theta(x-s_j,\tau)\Big)
 = \on{const} e^{2\pi i \mu x}\prod_{a=1}^{2m}\theta(x-z_a,\tau).
\eean
In particular, if $s_1,\dots,s_m$ 
are distinct  modulo the lattice $\Z+\tau\Z$ and 
any $s_j$ is distinct from any $z_a$ modulo the lattice $\Z+\tau\Z$, then
$(-\mu, s_1,\dots, s_m$, $z_1,\dots,z_{2m})$ is a solution of the Bethe ansatz equations \Ref{BAEs}.

\end{cor}

\begin{proof} 
The corollary follows from formula \Ref{h^2}.
\end{proof}

\subsection{Equivalence classes}
\label{sec Eqv}

\begin{lem} 
\label{lem tr}
Let
$(\mu, t_1,\dots, t_m,z_1,\dots,z_{2m})$
 be a solution of the Bethe ansatz equations \Ref{BAEs}. Then

\begin{enumerate}

\item[(i)]   for any
$j\in\{1,\dots,m\}$  the points $(\mu\mp 2, t_1,\dots,t_j\pm \tau,\dots,t_m,z_1,\dots,z_{2m})$ and
$(\mu, t_1$, \dots, $t_j\pm 1,\dots,t_m,z_1,\dots,z_{2m})$
 are also solutions of the Bethe ansatz equations;

\item[(ii)] 
for any
$k\in\{1,\dots, 2m\}$ the points $(\mu\pm1,t_1,\dots,t_m, z_1,\dots,z_k\pm\tau,\dots,z_{2m})$ 
and
$(\mu,t_1,\dots,t_m, z_1,\dots,z_k\pm1,\dots,z_{2m})$ are 
also solutions of the Bethe ansatz equations.
\end{enumerate}

\qed
\end{lem}

\begin{defn}
We say that  two solutions  
\bea
(\mu, t_1, \dots,  t_m ,  z_1 , \dots,  z_{2m}),\
(\mu', t_1',\dots, t_m', z_1',\dots,z_{2m}')  \ 
\in \C\times \C^m/S_m\times \C^{2m}/S_{2m}
\eea
 of the Bethe ansatz equations \Ref{BAEs}
 are equivalent if one of them is obtained from the other by a sequence of transformations listed in Lemma \ref{lem tr}.

\end{defn}

Denote by $\on{Sol}_m^1$ the set of equivalence classes of solutions
$ (\mu, t_1,\dots,t_m,z_1,\dots,z_{2m})$ of the Bethe ansatz equations
 \Ref{BAEs} with $ \mu\not\in \Z$.

\begin{thm}
\label{lem:sum}

We have 
 a bijective correspondence
\bea
\beta \ : \  \on{Sol}_m^1 \ \to\  \on{Pairs}_m^1
\eea
between the set  $\on{Sol}_m^1$ 
of equivalence classes of solutions of the Bethe ansatz equations
with $\mu\not\in \Z$  and the set $\on{Pairs}_m^1$ 
of points of $(P_{m,1}\times (P_m-P_{m,1}))/\Z$ with generic first coordinates. 
The correspondence is given below in the proof of this lemma.

\end{thm}

\begin{proof}
Take an equivalence class of solutions of the Bethe ansatz equations \Ref{BAEs} with $\mu\not\in\Z$ and choose a representative
$(\mu, t_1,\dots, t_m,z_1,\dots,z_{2m})$. Then let
$f(x) =\prod_{j=1}^m\theta(x-t_j,\tau)$ and let the function $g(x)$ be given by formula \Ref{nft}.
Then the pair $(f,g)$ determines a point of $(P_{m,1}\times (P_m-P_{m,1}))/\Z$ 
 with generic first coordinate. 
Moreover, the equivalent solutions of the Bethe ansatz equations
 correspond under this construction to the pairs of theta-polynomials
 lying in  the $\Z$-orbit of $(f,g)$ in $P_{m,1}\times (P_m-P_{m,1})$.
 
 Indeed,  take, for example,  the solution
 $(\mu-2, t_1+\tau, t_2,\dots,t_m,z_1,\dots,z_{2m})$. Then
 \bea
 \tilde f(x) = \theta(x-t_1-\tau,\tau){\prod}_{j=2}^m\theta(x-t_j,\tau)
 =e^{2\pi i x}{\prod}_{j=1}^m\theta(x-t_j,\tau)
 \eea
 and $\tilde g(x)$ is determined from the equation
\bea
\Wr(e^{2\pi i x}{\prod}_{j=1}^m\theta(x-t_j,\tau), \tilde g(x)) = 
e^{-2\pi i (\mu-2) x}{\prod}_{a=1}^{2m}\theta(x-z_a,\tau).
\eea
Hence $\tilde g(x) = e^{2\pi i x} g(x)$ and the pairs
$(f,g)$ and $(\tilde f,\tilde g)$ lie in the same $\Z$-orbit.

Similarly  take the solution 
 $(\mu+1, t_1, \dots,t_m,z_1+\tau,z_2,\dots,z_{2m})$,
 then $\tilde f= f$ and $\tilde g$ is determined from the equation
 \bea
 \Wr (f,\tilde g) = e^{-2\pi i (\mu+1)x}
\theta(x-z_1-\tau){\prod}_{a=2}^{2m}\theta(x-z_a,\tau)
=
e^{-2\pi i \mu x}
{\prod}_{a=1}^{2m}\theta(x-z_a,\tau).
\eea
Hence $\tilde g=g$ and  $(f,g)= (\tilde f,\tilde g)$.

\smallskip

Conversely let  $p \in (P_{m,1}\times P_m-P_{m,1})/\Z$ 
be a point with generic first coordinate. Choose a representative $(f,g)$  and write
\bean
\label{repn}
f(x) ={\prod}_{j=1}^m\theta(x-t_j,\tau),\qquad
\Wr(f,g) = e^{-2\pi i \mu x}{\prod}_{a=1}^{2m}\theta(x-z_a,\tau),
\eean
for some $(\mu,t_1,\dots, t_m, z_1,\dots,z_{2m},\tau)$, $\mu\not\in\Z$. Then
the function $\Wr(f,g)/f^2=(g/f)'$ has zero residues and  
$(\mu,t_1,\dots, t_m, z_1,\dots,z_{2m},\tau)$ is a solution of the Bethe ansatz equations
\Ref{BAEs} by Lemma \ref{lem:res}.  Different choices of the representative $(f,g)$ 
of the point $p$ or different choices
of the presentations in \Ref{repn} for the functions $f$ and $\Wr(f,g)$ 
give equivalent solutions of the Bethe ansatz equations.

Indeed take, for example,  another presentation
\bea
f(x) ={\prod}_{j=1}^m\theta(x-t_j,\tau)
={\prod}_{j=1}^m\theta(x-t_j',\tau)
\eea
in \Ref{repn}. 
Then the solutions $(\mu,t,z)$ and $(\mu, t',z)$ are equivalent by Lemma \ref{lem 1.2}.
Or take another presentation
\bea
 e^{-2\pi i \mu x}{\prod}_{a=1}^{2m}\theta(x-z_a,\tau)
 = e^{-2\pi i \mu' x}{\prod}_{a=1}^{2m}\theta(x-z_a',\tau)
 \eea
in \Ref{repn}. 
Then the solutions $(\mu,t,z)$ and $(\mu', t,z')$ are equivalent by Lemma \ref{lem 1.2}.
Finally, take another representative
$(\tilde f,\tilde g)=(e^{2\pi i kx} f(x), e^{2\pi i kx} g(x))$ of the point $p$. Then
\bea
\tilde f(x) =\theta(x-t_1-k\tau ,\tau){\prod}_{j=2}^m\theta(x-t_j,\tau),\qquad
\Wr(\tilde f, \tilde g) = e^{-2\pi i (\mu-2k) x}{\prod}_{a=1}^{2m}\theta(x-z_a,\tau).
\eea
This gives a solution $(\mu-2k, t_1+k\tau,t_1,\dots,t_m,z_1,\dots,z_{2m})$, which is in the same equivalence class. 
\end{proof}

By Theorem \ref{lem:sum} we have a bijection
$\beta  :   \on{Sol}_m^1  \to \on{Pairs}_m^1$.
We also have a subset $\on{Pairs}_m \subset \on{Pairs}_m^1$. Denote
\bean
\label{Sol_m}
\on{Sol}_m = \beta^{-1}(\on{Pairs}_m).
\eean

\subsection{Fundamental differential operators}

In formula \Ref{fund t} we introduced the fundamental differential operator
$\D_{(\mu,t,z)}$ of a solution
$(\mu,t,z)$ of the Bethe ansatz equations.
The operator $\D_{(\mu,t,z)}$  is defined  by the eigenvalues
of the dynamical elliptic Bethe algebra on the eigenfunction $\Psi(\la_{12}, \mu, t,   z, \tau)$,
see \Ref{EIg}.
In Section \ref{sec fp} we introduced the fundamental differential operator $\D_{(f,g)}$ of a point
$(f,g) \in (P_{m,1}\times (P_m-P_{m,1}))/\Z$ with generic first coordinate. In Theorem \ref{lem:sum}
we established a bijection
$\beta  :   \on{Sol}_m^1 \to  \on{Pairs}_m^1$
 between the set of equivalence classes of solutions 
$(\mu,t,z)$ of the Bethe ansatz equations
with $\mu\not\in \Z$ and the set of points of  $(P_{m,1}\times (P_m-P_{m,1}))/\Z$ with generic first coordinates. Thus a solution
$(\mu,t,z)$ with $\mu\not\in \Z$  gets two fundamental differential operators: 
$\D_{(\mu,t,z)}$  and $\D_{(f,g)}$.

\begin{thm}
\label{lem fo=}

We have $\D_{(\mu,t,z)} = \D_{(f,g)}$.

\end{thm}

\begin{proof} 
 By Theorem \ref{prop fund diff op},
\bea
\D_{(\mu,t,z)}=\big(\der_x + (\ln u)\rq{}\big)\big(\der_x - (\ln u)\rq{}\big)
\eea
where $u(x)$ is given by   \Ref{def u}, while  
\bea
\D_{(f,g)} = \big(\der_x + (\ln u_1)\rq{}\big)\big(\der_x - (\ln u_1)\rq{}\big),
\eea
where $u_1(x)$ is given by \Ref{u12}. We have  $u=u_1$ by the correspondence described in
Theorem \ref{lem:sum}.
\end{proof}

\begin{cor} 
\label{cor e2fo}

 The fundamental differential operators of equivalent solutions
 $(\mu,t,z)$, $(\mu\rq{}$, $t\rq{},z\rq{})$
with $\mu,\mu\rq{}\not\in\Z$ are equal.
\end{cor}

\subsection{Analytic involution and Bethe ansatz}
\label{sec aiba}

Consider the analytic involution 
\bea
\iota_{\on{an}}: (P_{m,1}\times (P_m-P_{m,1}))/\Z  \to (P_{m,1}\times (P_m-P_{m,1}))/\Z 
\eea
defined in \Ref{an inv}.  
By its construction in Section \ref{sec ani},  the analytic involution 
restricts to an  involution
\bean
\label{Geni}
\iota_{\on{an}} \  :\  \on{Pairs}_m \  \to\  \on{Pairs}_m
\eean 
on the set $\on{Pairs}_m$ and induces an involution
\bean
\label{GenS}
\beta^{-1} \circ\iota_{\on{an}}\circ\beta
\ :\ 
  \on{Sol}_m \  \to\  \on{Sol}_m 
\eean 
on the set $\on{Sol}_m$.

In other words, the involution in \Ref{GenS} is defined as follows.
We start with a solution
$(\mu, t_1,\dots,t_m, z_1,\dots, z_{2m})$, $\mu\not\in\Z$, of the Bethe ansatz equations \Ref{BAEs},
construct the function $f=\prod_{j=1}^m\theta(x-t_j,\tau)$,
determine the function $g=e^{-2\pi i \mu x} \prod_{j=1}^m\theta(x-s_j,\tau)$ from the equation
$\Wr(f,g) = e^{-2\pi i \mu x } \prod_{a=1}^{2m}\theta (x-z_a,\tau)$. Then the involution in 
\Ref{GenS} sends the equivalence class of 
$(\mu, t_1,\dots,t_m, z_1,\dots, z_{rm})$ to the equivalence class
of $(-\mu, s_1$, \dots, $s_m$,  $z_1$,
\dots, $z_{2m})$, see Lemma \ref{lem 1.11} and Corollary
\ref{cor 9.4}.

\subsection{Normal solutions}

\begin{defn}
Given a fundamental parallelogram $\La$ of the lattice $\Z+\tau\Z$ acting on $\C$, a solution
\bea
(\mu, t_1, \dots,  t_m ,  z_1 , \dots,  z_{2m})\
\in \C\times \C^m/S_m\times \C^{2m}/S_{2m}
\eea
 of the Bethe ansatz equations \Ref{BAEs}
will be called $\on{ normal\ relative\ to}$  $\La$  if 
$(t_1, \dots,  t_m)\in \La^m/S_m$ and $(z_1 , \dots,  z_{2m})
\in  \La^{2m}/S_{2m}$.

\end{defn}

\begin{lem}
\label{lem eqv}
Given a fundamental parallelogram $\La$, then every equivalence class of solutions of the Bethe ansatz equations \Ref{BAEs}
has a unique normal solution.
\qed

\end{lem}

\section{Bethe eigenfunctions for $(\C^2)^{\ox n}[0]$}
\label{sec BAf}

Notice that by Corollary \ref{cor Bet alg conjug} the dynamical elliptic Bethe algebra
$\B^V(z_1,\dots,z_{2m},\tau)$  does not change up to conjugation if the numbers
$z_1,\dots,z_{2m}$ are shifted by elements of the lattice $\Z+\tau\Z$.
We choose a fundamental parallelogram $\La\subset \C$ of the lattiuce $\Z+\tau\Z$
and from now on 
 {\it fix  $z_1,\dots, z_{2m} $ 
to be distinct points of the open parallelogram} $\La\rq{}\subset \La$.

\subsection{Formula for eigenfunctions}

Let $v_1, v_2$ be the standard basis of the $\slt$-module $\C^2$,
 \bea
&&
 (e_{11}-e_{22})v_1=v_1, 
 \quad   e_{21} v_1 = v_2
 \quad
 e_{12}v_1=0,
\\
&&
(e_{11}-e_{22})v_2=-v_2,
\quad
 e_{21} v_2 = 0,
 \quad
 e_{12} v_2 = v_1.
 \eea
    A basis $(v_I)_I$ of $V=(\C^2)^{\ox n}$ is labeled by subsets
 $I\subset \{1,\dots,2m\}$, where
 \bea
 v_I = v_{j_1}\ox \dots \ox v_{j_{2m}},
 \eea
 with $j_i = 2$ if $i\in I$ and $j_i=1$ if $i\in\bar I$. The 
vectors $(v_I)_{|I|=m}$ form a 
 basis of the zero weight subspace
 $V[0]$.

For $(\mu, t_1,\dots,t_m, z_1,\dots,z_{2m})\in \C^{3m+1}$
define a $V[0]$-valued function $\Psi(\la_{12}, \mu, t,  z,\tau)$
by the formula:
\bean
\label{eq eig}
\Psi(\la_{12}, \mu, t,  z,\tau) = e^{\pi i \mu \la_{12}} {\sum}_{|I|=m} W_I(\la_{12}, t, z,\tau)\, v_I,
\eean
where for $I=\{i_1<\dots <i_m\}\subset \{1,\dots,2m\}$ we define
\bean
\label{prod}
W_I(\la_{12},t,z) = \on{Sym}_{t_1,\dots,t_m} \Big( {\prod}_{j=1}^m \si(t_j-z_{i_j},-\la_{12},\tau)\Big)
\eean
and
$
 \on{Sym}_{t_1,\dots,t_m} \left(F(t_1,\dots,t_m)\right) =
\sum_{\si\in S_m} F(t_{\si(1)},\dots,t_{\si(m)}).
$

\smallskip
We have the following periodicity property:
\bean
\label{Flocke}
\Psi(\la_{12}+1, \mu, t,  z,\tau)\, = \,e^{\pi i \mu} \,\Psi(\la_{12}, \mu, t,  z,\tau).
\eean

\bigskip
In the considered case of  $V=(\C^2)^{\ox n}$, Theorem \ref{FV thm} takes the following form.

\begin{thm} [\cite{FV1}]
\label{thm eig}

Let $t$ be a solution of the Bethe ansatz equations \Ref{BAEs}. Then the $V[0]$-valued 
function
$\Psi(\la_{12}, \mu, t,  z,\tau)$ of $\la_{12}$, defined in \Ref{eq eig}, is such that
\bea
H_a(z,\tau)\, \Psi(\la_{12}, \mu, t,  z,\tau) 
&=&
 \frac {\der \Phi}{\der z_a}(\mu, t, z,\tau) \,\Psi(\la_{12}, \mu, t,  z,\tau), \qquad a=1,\dots,2m,
\\
H_0(z,\tau) \Psi(\la_{12}, \mu, t,  z,\tau) 
&=&
 \frac {\der \Phi}{\der \tau}(\mu, t, z,\tau)\, \Psi(\la_{12}, \mu, t,  z,\tau).
\eea

\end{thm}

\subsection{ Eigenfunctions of equivalent solutions}

In Section \ref{sec Eqv} we introduced the notion of equivalent solutions of the Bethe ansatz equations \Ref{BAEs}

\begin{thm} 
\label{lem equ}

Let $(\mu,t,z,\tau)$ and $(\mu',t',z,\tau)$ be two equivalent solutions of the Bethe ansatz equations \Ref{BAEs} with the same $z$, then $ \Psi(\la_{12}, \mu, t,  z,\tau)=
 \Psi(\la_{12}, \mu', t',  z,\tau)$.

\end{thm}

\begin{proof} Let $(\mu,t,z,\tau)$ be a solution. Consider an equivalent solution
$(\mu-2, t_1,\dots,t_k+\tau,\dots, t_m, z,\tau)$, for some $k$, $1\leq k\leq m$. We  show that the corresponding 
eigenfunctions are equal. Indeed the common factor $e^{\pi i \mu\la_{12}}$ in \Ref{eq eig}
is transformed into $e^{\pi i \mu\la_{12}}e^{-2\pi i \la_{12}}$, while
the factor  $\si(t_k-z_i, -\la_{12},\tau)$ in the product in \Ref{prod} is transformed into
$\si(t_k+\tau-z_i, -\la_{12},\tau)=e^{2\pi i \la_{12}}\si(t_k-z_i, -\la_{12},\tau)$. The
two new factors are canceled and the corresponding eigenfunctions are equal.  This proves the theorem.
\end{proof}

\subsection{Eigenfunctions and two involutions}

Let $\Psi(\la_{12})$ be a $V[0]$-valued eigenfunction of the dynamical elliptic Bethe algebra
$\B^V(z_1,\dots,z_{2m})$  and with the fundamental
differential operator $\D_\Psi$, see \Ref{fund Psi}.
 Recall the Weyl group of $\slt$, $W=\{\on{id}, s\}$. 
We have
\bea
\D_{s(\Psi)} = \D_\Psi,
\eea 
since $\B^V(z_1,\dots,z_{2m})$ is Weyl group invariant by Lemma \ref{weyl inv}.
On the other hand, let $(\mu,t,z)\in \on{Sol}_m$ be an equivalence class of
solutions of the Bethe ansatz equation \Ref{BAEs}, see \Ref{Sol_m}.
Let $\Psi(\la_{12},\mu,t,z)$ be the associated eigenfunction of the dynamical
elliptic Bethe algebra
$\B^V(z_1,\dots,z_{2m})$
with the fundamental differential operator  $\D_{(\mu,t,z)}$, 
see  \Ref{fund t}. Both $\Psi(\la_{12},\mu,t,z)$ and  $\D_{(\mu,t,z)}$ 
are well-defined, see Theorem  \ref{lem equ}.

 Recall the analytic involution
\bea
\beta^{-1} \circ\iota_{\on{an}}\circ\beta\ :\ \on{Sol}_m\ \to\ \on{Sol}_m,
\eea
defined in \Ref{GenS}. Let $(-\mu, s, z)$ be the image of 
$(\mu,t,z)$ under this involution, see Section \ref{sec aiba}.
Let $\Psi(\la_{12},-\mu,s,z)$ be the associated eigenfunction.
 Then
\bea
\D_{(-\mu,s,z)} = \D_{(\mu,t,z)},
\eea
by Theorem \ref{lem fodp}.  Now the following three eigenfunctions
have the same fundamental differential  operator: $\Psi(\la_{12},\mu,t,z)$,
$s(\Psi(\la_{12},\mu,t,z))$, $\Psi(\la_{12},-\mu,s,z)$.

\begin{thm}
\label{thm inv}
We have
\bea
s(\Psi(\la_{12},\mu,t,z)) = \on{const}\ \Psi(\la_{12},-\mu,s,z).
\eea
In other words, the Weyl involution coincides with the analytic involution.

\end{thm}

The theorem is proved in Section \ref{end}.

\begin{thm}
\label{thm BdE} 

 Assume that $ (\mu,t,z), (\mu\rq{},t\rq{},z) \in \on{Sol}_m$ and
the two Bethe eigenfunctions  $\Psi(\la_{12},\mu,t,z)$ and
 $\Psi(\la_{12},\mu\rq{},t\rq{},z)$ have the same eigenvalues for every element of the dynamical
elliptic Bethe algebra $\B^V(z_1,\dots,z_{2m})$. Then
either $ (\mu,t,z) = (\mu\rq{},t\rq{},z) $ or $ (\mu,t,z)$ is the image of 
$ (\mu\rq{},t\rq{},z)$ under the analytic involution  $\beta^{-1} \circ\iota_{\on{an}}\circ\beta$.

\end{thm}

Theorem \ref{thm BdE}  says that the dynamical elliptic Bethe algebra separates the Weyl group orbits
of the Bethe eigenfunctions.

\begin{proof} The assumptions of the theorem mean that $\Psi(\la_{12},\mu,t,z)$ and
 $\Psi(\la_{12},\mu\rq{},t\rq{},z)$ have the same fundamental differential operators.
Now the  theorem follows from Theorem \ref{lem fodp}.
\end{proof}

\section{Elliptic Wronski map}
\label{sec elWm}

\subsection{Wronski map}
Define the {\it elliptic Wronski map}
\bean
\label{WR}
 P_m\times P_m\to P_{2m},
\qquad
(f,g) \mapsto \Wr(f(x),g(x)).
\eean
The group $\C$ acts on $P_m\times P_m$ 
by the operators $L_\nu\times L_\nu$, $\nu\in\C$.
The maps $L_{2\nu}$, $\nu\in\C$, define an action of $\C$ on $P_{2m}$.
The Wronski map commutes with the actions of $\C$ on $P_m\times P_m$ and $P_{2m}$: 
\bean
\label{eqW}
\Wr \circ \,(L_\nu\times L_\nu) =L_{2\nu}\circ \Wr,
\eean
see formula \Ref{h^2},

We are interested in the {elliptic  Wronski map}
\bean
\label{Wr}
\Wr : P_{m,1}\times (P_m-P_{m,1}) \ \to \ P_{2m}-P_{2m,1}.
\eean
 This is a holomorphic map between  $2m+1$-dimensional complex manifolds.

Notice that the map in \Ref{Wr} induces  a map 
\bean
\label{Wr Z}
 (P_{m,1}\times (P_m-P_{m,1}))/\Z \ \to \ (P_{2m}-P_{2m,1})/\Z,
\eean
see \Ref{eqW}.

\subsection{Wronski map has nonzero Jacobian} 

\begin{lem}
\label{lem Jac}
The Jacobian of the  Wronski map $\Wr : P_{m,1}\times (P_m-P_{m,1}) \to  P_{2m}-P_{2m,1}$
is nonzero at generic points of $P_{m,1}\times (P_m-P_{m,1})$.
\end{lem}

\begin{proof}
Let $(f,g)$ be a generic point, $f(x) = \prod_{j=1}^m\theta(x-t_j,\tau)$,
$g(x) = e^{-2\pi i \mu x}\prod_{j=1}^m\theta(x-s_j,\tau)$
with some distinct $t_1,\dots,t_m$, $s_1,\dots,s_m$ modulo the lattice $\Z+\tau\Z$. 
The parameters $\mu$, $t_1,\dots,t_m$, $s_1,\dots,s_m$ are local coordinates on 
$P_{m,1}\times (P_m-P_{m,1})$. 

Let  $e^{-2\pi i \mu x}\prod_{a=1}^{2m}\theta(x-z_a,\tau)\in P_{2m}-P_{2m,1}$.
The  parameters $\mu$, $z_1,\dots,z_{2m}$ are local coordinates on 
$P_{2m}-P_{2m,1}$.

In order to prove the lemma it is enough to prove that for  fixed generic $\mu$, the map
\bean
\label{fixed mu}
&&
\big({\prod}_{j=1}^m\theta(x-t_j,\tau),e^{-2\pi i \mu x}{\prod}_{j=1}^m\theta(x-s_j,\tau)\big)
\\
\notag
&&
\phantom{aaaaaaaaaa}
\mapsto \ \
\Wr(f,g)=e^{-2\pi i \mu x}{\prod}_{a=1}^{2m}\theta(x-z_a,\tau),
\eean
sending 
$(t_1,\dots$, $t_m, s_1,\dots,s_m)$ to   $(z_1,\dots,z_{2m})$ has nonzero Jacobian at
 a generic point $(t_1,\dots,t_m, s_1,\dots,s_m)$.

First consider the multiplication map
\bean
\label{fixe mu}
\phantom{aaaaa}
\big(\prod_{j=1}^m\theta(x-t_j,\tau),e^{-2\pi i \mu x}\prod_{j=1}^m\theta(x-s_j,\tau)\big)\mapsto
e^{-2\pi i \mu x}\prod_{j=1}^{m}\theta(x-t_j,\tau)\prod_{j=1}^m\theta(x-s_j,\tau)
\eean
sending 
$(t_1,\dots,t_m, s_1,\dots,s_m)$ to   $(z_1,\dots,z_{2m})=(t_1,\dots,t_m, s_1,\dots,s_m)$ 
which is a local isomorphism. For large $\mu$ the map in
\Ref{fixed mu} is a small deformation
of the map in \Ref{fixe mu}.
Indeed
\bea
&&
\Wr\big(\prod_{j=1}^m\theta(x-t_j,\tau),e^{-2\pi i \mu x}\prod_{j=1}^m\theta(x-s_j,\tau)\big)
\\
&&
=
-2\pi i \mu \,e^{-2\pi i \mu x}\Big(\prod_{j=1}^{m}\theta(x-t_j)\prod_{j=1}^m\theta(x-s_j)
- \frac 1{2\pi i \mu} \Wr(\prod_{k=1}^{j}\theta(x-t_j),\prod_{j=1}^m\theta(x-s_j))\Big).
\eea
Hence for large $\mu$ the Jacobian of the map in \Ref{fixed mu} is not identically equal to zero.
\end{proof}

Another proof of Lemma \ref{lem Jac} can extracted from the proof of Theorem
\ref{thm deg} below.

\subsection{Labels}
\label{sec lab}

Choose $w\in\C$. 
Let $\Lambda =\La_w\subset \C$ be the
fundamental parallelogram  with vertices $w,w+1,w+\tau,w+1+\tau$ and with boundary
intervals $[w,w+1), [w,w+\tau)$ included and  boundary intervals
$[w+1,w+1+\tau], [w+\tau,w+1+\tau]$ excluded. 
Denote by $\bar\La$ the closure of $\La$, by $\der\bar \La$ the boundary of
$\bar\La$ and by $\La\rq{}$ the open parallelogram $\bar\La-\der\bar\La$.
The parallelogram $\Lambda$ is a fundamental domain
for the $\Z+\tau\Z$-action on $\C$.

For any theta-polynomial $h\in P_k$ (considered up to multiplication by a nonzero constant)
there exist unique 
$u=(u_1,\dots, u_k) \in\La^k/S_k$ and $\mu\in\C$ such that
\bea
h(x) = e^{2\pi i \mu x}{\prod}_{j=1}^k \theta(x-u_j,\tau).
\eea
The pair $(u, \mu) \in \La^k/S_k\times \C$ will be called the {\it coordinates of }$h(x)$
relative to the fundamental  parallelogram $\La$.
The number $\mu$ will be called the {\it label} of $h$ relative to
$\La$ and denoted $l(h)$.

\subsection{Labeled preimage}
Recall the Wronski map 
\bea
\Wr : P_{m,1}\times (P_m-P_{m,1}) \ \to \ P_{2m}-P_{2m,1}.
\eea
For $h\in P_{2m}-P_{2m,1}$ and $k\in\Z$ define the {\it labeled preimage} of $h$ with label $k$
as the set
\bean
\label{rW}
\Wr_k^{-1}(h) := \{ (f,g)\in P_{m,1}\times (P_m-P_{m,1}) 
\ |\  \Wr(f,g) = h, \ l(f)=k\}.
\eean
We have
\bea
\Wr^{-1}(h) = \cup_{k\in\Z} \Wr_k^{-1}(h),
\eea
and $\Wr_k^{-1}(h)\cap \Wr_{k\rq{}}^{-1}(h)=\emptyset$ if $k\ne k\rq{}$.

\begin{thm}
\label{thm deg}

Choose a fundamental parallelogram $\La$ and  
$z=(z_1,\dots,z_{2m})\in (\Lambda\rq{})^{2m}/S_{2m}$ with distinct coordinates. 
Then there exists  $N >0$, such that for any 
$h\in P_{2m}-P_{2m,1}$ of the form
\bean
\label{h  mu}
h(x)= e^{-2\pi i\mu x}{\prod}_{a=1}^{2m}\theta(x-z_a,\tau),
\eean
 with $|\on{Im} \mu|>N$,
 and any $k\in\Z$, the set $\Wr^{-1}_k(h)$ consists of
 exactly $\binom{2m}{m}$ points.
Moreover, in that case any  $(f,g)\in \Wr^{-1}_k(h)$ has the form
\bean
\label{nf}
f(x) = e^{2\pi i kx}{\prod}_{j=1}^m\theta(x-t_j,\tau),\qquad
g(x) = e^{-2\pi i(\mu+ k)x}{\prod}_{j=1}^m\theta(x-s_j,\tau),
\eean
where $t=(t_1,\dots,t_m)\in(\La\rq{})^m/S_m$, $s=(s_1,\dots,s_m)\in(\La\rq{})^m/S_m$,
and all the numbers $t_1,\dots,t_m$, $s_1,\dots,s_m$, 
$z_1,\dots,z_{2m}$ are pairwise distinct.

\end{thm}

Theorem \ref{thm deg} is proved in Section \ref{proofs}.

\begin{cor}
\label{cor deg}

Under the assumptions of Theorem \ref{thm deg},  the 
points $(f,g)\in \Wr^{-1}_k(h)$ 
have generic  first and second coordinates,
and any two points of $\Wr^{-1}(h)$  
  lie in different $\Z$-orbits. Moreover, at each point $(f,g)\in \Wr^{-1}(h)$
the Jacobian of the Wronski map is nonzero.
\qed
\end{cor}

\begin{rem}

Some additional asymptotic structure of the points of $\Wr^{-1}_k(h)$ 
as $\on{Im} \mu \to \infty$
can be observed in the proof of Theorem \ref{thm deg}.

\end{rem}

\subsection{Proof of Theorem \ref{thm deg}}
\label{proofs}

For $\mu\in\C-\Z$ consider the multiplication map
\bea
M :  P_{m,1}\times P_{m,e^{-2\pi i \mu}} \to   P_{2m,e^{-2\pi i \mu}}, \quad (f,g)\mapsto fg,
\eea
and its fiber 
\bea
M^{-1}(h)= \{ (f,g) \in  P_{m,1}\times P_{m,e^{-2\pi i \mu}}\ |\
fg = h\}
\eea
over a  point $h=e^{-2\pi i\mu x}\prod _{j=1}^{2m}\theta(x-z_j)$ with 
$(z,\mu) \in (\Lambda\rq{})^{2m}/S_{2m}\times(\C-\Z)$ and  $z$ with distinct coordinates.
Then the points of $M^{-1}(h)$ are labeled by 
the pairs $(I,k)$, where 
 $I\subset\{1,\dots,2m\}$ is an $m$-element subset and 
$k\in\Z$. The corresponding point $(f_{I,k},g_{I,k})\in M^{-1}(h)$ has the form
\bea
f_{I,k}  = e^{2\pi i kx} {\prod}_{j\in I}\theta(x-z_j,\tau),
\qquad
g_{I,k}  = e^{-2\pi i (\mu+k) x} {\prod}_{j\in \bar I}\theta(x-z_j,\tau),
\eea
where $\bar I$ is the the complement to $I$ in $\{1,\dots,2m\}$. 

Recall $\bar\La \subset \C$, the closure of the parallelogram $\La$. For $k\in\Z$ denote
\bea
V_k =\{f \in P_{m,1}\ |\ f= e^{2\pi i kx}{\prod}_{j=1}^m\theta(x-t_j,\tau)\ 
\text{where} \ t=(t_1,\dots,t_m)\in \bar\La^m/S_m\},
\\
V_{\mu, k} =\{g\in P_{m,e^{-2\pi i \mu}} | \,
g=e^{-2\pi i (\mu +k) x}\prod_{j=1}^m\theta(x-s_j,\tau)\ \text{where} \ s=(s_1,\dots,s_m)\in \bar\La^m/S_m\}.
\eea
We see that for any $k_1,k_2\in\Z$ the sets $V_{k_1}, V_{\mu,k_2}, V_{k_1}\times V_{\mu,k_2}$ are compact subsets of 
$P_{m,1}$, 
$P_{m, e^{-2\pi i \mu}}$,  $P_{m,1}\times
P_{m, e^{-2\pi i \mu}}$, respectively,  and 
\bea
  P_{m,1} = \cup_{k_1\in\Z}  V_{k_1},
\qquad
P_{m, e^{-2\pi i \mu}} = \cup_{k_2\in\Z} V_{\mu,k_2}\qquad
  P_{m,1}\times
P_{m, e^{-2\pi i \mu}} = \cup_{k_1,k_2\in\Z} V_{k_1}\times V_{\mu,k_2}.
\eea 
The intersection $M^{-1}(h)\cap (V_{k_1}\times V_{\mu,k_2})$
consists of $\binom{2m}{m}$ points if $k_1=k_2$ and is empty if $k_1\ne k_2$.

Consider the Wronski map
\bea
\Wr :  P_{m,1}\times P_{m,e^{-2\pi i \mu}} \to   P_{2m,e^{-2\pi i \mu}}, 
\quad (f,g)\mapsto \Wr(f,g),
\eea
and the fiber 
$\Wr^{-1}(h)= \{ (f,g) \in  P_{m,1}\times P_{m,e^{-2\pi i \mu}} |\
\Wr(f,g) = h\}$
over the same point $h=e^{-2\pi i\mu x}\prod _{j=1}^{2m}\theta(x-z_j)$. 

\medskip
Let  $(f,g)\in V_{k_1}\times V_{\mu,k_2}$ with $ f= e^{2\pi i k_1x}\prod_{j=1}^m\theta(x-t_j)$
for some $(t_1,\dots,t_j)\in \bar\La^m/S_m$ and 
$g=e^{-2\pi i (\mu +k_2) x}\prod_{j=1}^m\theta(x-s_j)$ for some $(s_1,\dots,s_m)\in \bar\La^m/S_m$.
Then
\bean
\label{app}
&&
\phantom{aaa}
\Wr(f,g)= 
-2\pi i (\mu+k_2+k_1) \,e^{-2\pi i (\mu+k_2-k_1) x}
\\
\notag
&&
\phantom{aaa}
\times
\Big(\prod_{j=1}^{m}\theta(x-t_j)\prod_{j=1}^m\theta(x-s_j)
- \frac 1{2\pi i (\mu+k_1+k_2)} \Wr\big(\prod_{j=1}^{m}\theta(x-t_j),
\prod_{j=1}^m\theta(x-s_j)\big)\Big).
\eean
Since  the functions $f, g, \Wr(f,g)$ are considered up to multiplication by nonzero numbers, we may ignore the 
first
factor $-2\pi i (\mu+k_2+k_1)$ in the right-hand side.

Let us analyze the last factor in \Ref{app}.  Let
\bean
\label{exx}
&&
G(x,t,s,v)=\prod_{j=1}^{m}\theta(x-t_j)\prod_{j=1}^m\theta(x-s_j)
- v \Wr\big(\prod_{j=1}^{m}\theta(x-t_j),\prod_{j=1}^m\theta(x-s_j)\big).
\eean
The function $G(x,t,s,v)$, as a function of $x$, has the same multipliers as 
the function
$\prod_{j=1}^{m}\theta(x-t_j)\prod_{j=1}^m\theta(x-s_j)$, see the proof of Lemma \ref{lem 1.5}.
Namely,  the first multiplier of  $G(x,t,s,v)$
is 1 and the second is $e^{2\pi i \sum_{j=1}^m(t_j+s_j)}$.
By Lemma  \ref{lem determin}, for any $(t,s,v)$  we have 
\bean
\label{Fcu}
G(x,t,s,v) = c(t,s,v)\prod_{j=1}^{2m}\theta(x-u_j(t,s,v)) 
\eean
for some $c(t,s,v)\in \C$, $u(t,s,v)$ $=$ $(u_1(t,s,v)$, \dots, $u_{2m}(t,s,v)) \in \C^{2m}/S_{2m}$
and
\bean
\label{mult eq}
e^{2\pi i \sum_{j=1}^{2m}u_j(t,s,v)}=e^{2\pi i \sum_{j=1}^m(t_j+s_j)}.
\eean
The pair  $(c(t,s,v), u(t,s,v))$ is not unique, see Lemma \ref{lem 1.2}.

\begin{lem}
\label{lem hol}
Consider $(\La\rq{})^{2m}$ with coordinates $t_1,\dots,t_m,s_1,\dots,s_m$.
Let $C$ be a compact subset of $(\La\rq{})^{2m}$ which is disjoint from all diagonals and invariant with respect to the $S_m\times S_m$-action.
Then there exists $\delta>0$, such that the pair  $(c(t,s,v), u(t,s,v))$ in \Ref{Fcu}
can be chosen so that 
\bean
\label{CCm}
\phantom{aaaa}
\C^m/S_m\times \C^m/S_m\times \C \to \C\times\C^{2m}/S_{2m},
\qquad
(t,s,v)\to  (c(t,s,v), u(t,s,v)),
\eean
  is a  well-defined  holomorphic map for $(t,s)\in C/(S_m\times S_m)$ and $v\in\C$, $|v|<\delta$, and
\begin{enumerate}
\item[(i)] $u(t,s,v)$ $=$ $(u_1(t,s,v)$, \dots, $u_{2m}(t,s,v))$ has distinct coordinates;

\item[(ii)]
 $(c(t,s,0), u(t,s,0))=(1,(t,s))$;

\item[(iii)] for any $v$ with $|v|<\delta$ the restriction map
\bean
\label{rmap}
C/(S_m\times S_m)\times \{v\} \to \C^{2m}/S_{2m},\qquad
(t,s,v)\to   u(t,s,v),
\eean
has nonzero Jacobian.

\end{enumerate}
\end{lem}

\medskip
\noindent
{\it Proof of Lemma \ref{lem hol}.}
We first construct a certain holomorphic map 
\bean \label{hol rts}
C \times  \{v\in\C\ |\ |v|\ll 1\} \to \C
\times\C^{2m},\qquad ((t,s),v) \mapsto (\tilde c(t,s,v), \tilde u(t,s,v)),
\eean
as follows.

For $(t,s) \in C$,  the numbers $x=t_j$, $j=1,\dots,m$, and $x=s_j$, $j=1,\dots,m$,
are simple zeros of the function $G(x,t,s,0)$.
By  the implicit function theorem there exist unique holomorphic functions $\tilde u_j(t,s,v)$, $j=1,\dots,m$,
 defined in a neighborhood of $(t,s,0)$ and such that $G(\tilde u_j(t,s,v),t,s,v)=0$, $\tilde  u_j(t,s,0) = t_j$.  Similarly, 
 there exist unique holomorphic functions $\tilde  u_j(t,s,v)$, $j=m+1,\dots,2m$,
 defined in a neighborhood of $(t,s,0)$ and such that $G(\tilde  u_j(t,s,v),t,s,v)=0$,  
$\tilde u_j(t,s,0) = s_{j-m}$.

Since $C$ is compact, there is a $\delta >0$ such that 
$\tilde u(t,s,v)=(\tilde u_1(t,s,v),\dots, \tilde u_{2m}(t,s,v))$ is holomorphic for $(t,s)\in
 C$ and $v\in\C$, $|v|<\delta$, and has distinct coordinates.
Denote $G(x,t,s,v) = \prod_{j=1}^{2m}\theta(x-\tilde  u_j(t,s,v))$ and define
$\tilde c(x, t,s,v)$ by the formula
\bean
\label{FG}
G(x, t,s,v) =\tilde  c(x,t,s,v) H(x,t,s,v).
\eean
Since $G(x, t,s,v)$ and $H(x,t,s,v)$ have the same zeros, the function $\tilde  c(x,t,s,v)$ 
is holomorphic in $x,t,s,v$ and has the from
\bean
\label{FbG}
\tilde c(x,t,s,v) = \hat c(t,s,v) e^{2\pi i k(t,s,v)x}
\eean
 for some $\hat c(t,s,v)\in\C$, $k(t,s,v)\in\Z$, such that $k(t,s,0)=0$ and $\hat c(t,s,0)=1$.
 We know that the second multiplier of the  function $ G(x,t,s,v)$, as a function of $x$,  is
$e^{2\pi i \sum_{j=1}^m(t_j+s_j)}$.  From \Ref{FbG}
we conclude that the second multiplier of $G(x,t,s,v)$ is
$e^{2\pi i k(t,s,v)\tau + 2\pi i \sum_{j=1}^{2m} u_j(t,s,v)}$. 
Since $ \sum_{j=1}^{2m}u_j(t,s,v)$ is continuous, we conclude that $k(t,s,u)=0$ and
$\tilde c(x,t,s,v)$ in 
\Ref{FG} does not depend on $x$. The map in \Ref{hol rts} is constructed. Clearly,
we can choose the positive $\delta$ so small that for any $v$ with $|v|<\delta$, the restriction map
\bea
C \times  \{v\} \to 
\C^{2m},\qquad ((t,s),v) \mapsto \tilde u(t,s,v),
\eea
has nonzero Jacobian, since it is a deformation of the identity map.

For any permutations $\si, \eta \in 
S_m$ and $j=1,\dots,m$, we clearly  have
\bea
&&
\tilde u_{j}(t_{\si_1}, \dots, t_{\si_m}, s_{\eta_1}, \dots, s_{\eta_m},v)=
\tilde u_{\si_j}(t_1,\dots,t_m,s_1,\dots,s_m,v),
\\
&&
\tilde u_{m+j}(t_{\si_1}, \dots, t_{\si_m}, s_{\eta_1}, \dots, s_{\eta_m},v)=
\tilde u_{m+\eta_j}(t_1,\dots,t_m,s_1,\dots,s_m,v).
\\
&&
\tilde c(t_{\si_1}, \dots, t_{\si_m}, s_{\eta_1}, \dots, s_{\eta_m},v)=
\tilde c(t_1,\dots,t_m,s_1,\dots,s_m,v).
\eea
Hence the map in \Ref{hol rts} projects to  the required  map in \Ref{CCm}
after factorizing the preimage of the map in \Ref{hol rts} by $S_m\times S_m$.
 The lemma is proved.
\qed

\medskip

Recall our $z=(z_1,\dots,z_{2m})\in (\La\rq{})^{2m}/S_{2m}$ with distinct coordinates. 
Let $\ep_1>0$ be the distance from the set $\{z_1$, \dots, $z_{2m}\}$ $\subset \La\rq{}$ to the boundary $\der \bar\La$
 of $\bar\La$. Let $\ep_2=\on{min}_{ a, b,\, 1\leq a <b\leq 2m}|z_a-z_b|$. 
Choose any $\ep$ with  $0<\ep<\on{min}\{\ep_1/2,\ep_2/2\}$.

Denote by $Z$ the set of all points $(t,s)=(t_1,\dots,t_m,s_1,\dots,s_m)\in\bar\La^{2m}$,
 which have at least two equal coordinates
or at least one coordinate lying in $\der\bar\La$. Clearly, $Z$ is closed and invariant with
 respect to the $S_{m}\times S_m$-action on
$\bar\La^{2m}$.

\begin{lem}
\label{lem der} 

There  exists a neighborhood $U$ of $Z$ in $\bar\La^{2m}$ and a number $\delta\rq{}>0$ 
with the following properties:

\begin{enumerate}
\item[(i)]  $U$ is $S_{m}\times S_m$-invariant;

\item[(ii)] the point $z$ does not lie in the closure $\overline{U/S_{2m}}$ of $U/S_{2m}$;

\item[(iii)] for any $(t,s)\in U$, $v\in\C$ with $|v|<\delta\rq{}$, the function $F(x,t,s,v)$ has a multiple zero, or has
a zero with the distance less than $\ep$ to the boundary $\der\bar\La$, or has two distinct zeros
 with the distance between them less than $\ep$.

\end{enumerate}

\end{lem}

\begin{proof}
The lemma clearly follows from the fact that $\bar\La^{2m}$ is compact and $F(x,t,s,v)$ is holomorphic.
\end{proof}

Let $U,\delta\rq{}$ be as in Lemma \ref{lem der}. 
Define the compact set $C=\bar\La^{2m}-U$. The set $C$ satisfies the
assumptions of Lemma \ref{lem hol}.
By Lemma \ref{lem hol} there exists $\delta$,  $0<\delta<\delta\rq{}$, such that 
the function $u(t,s,v)$ $=$ $(u_1(t,s,v)$, \dots, $u_{2m}(t,s,v))$ of Lemma \ref{lem hol} has the following properties:
\smallskip

\begin{enumerate}
\item[(iv)]

For any $v$ with $|v|<\delta$, the equation $u(t,s,v) = z$ has exactly $\binom{2m}{m}$
distinct  solutions 
$(t(v),s(v)) \in C/(S_m\times S_m$.  Each of the solutions is a holomorphic function of $v$.
The solutions are labeled by  $m$-element subsets   $I\subset\{1,\dots,2m\}$. The corresponding solution
$(t_I(v),s_I(v))$ is such that 
\bean
t_I(0)=(z_j)_{j\in I}, \qquad
s_I(0)=(z_j)_{j\in \bar I}.
\eean

\item[(v)] For $(t,s)\in C/(S_m\times S_m)$ and $v$ with $|v|<\delta$, each coordinate of
$u(t,s,v)$ lies in the $\ep$-neighborhood of the parallelogram $\bar\La$.

\end{enumerate}

\medskip

Let us return to the proof of Theorem \ref{thm deg}.
Let $h=e^{-2\pi i \mu x}\prod_{a=1}^{2m}\theta(x-z_a)$ as before.
Assume that $\mu\in\C$ is such that $|\on{Im} \mu| > \delta^{-1}$, then 
$|2\pi i(\mu+k_2+k_1)|^{-1}<\delta$ for any $k_1,k_2\in\Z$.
Then formula \Ref{app}, property (iii) of Lemma \ref{lem der}
and  properties (iv-v) above show that the intersection 
$\Wr^{-1}(h)\cap (V_{k_1}\times V_{\mu,k_2})$ 
is empty if $k_1\ne k_2$ and consists of exactly  $\binom{2m}{m}$ 
points if $k_1=k_2$, moreover, those points  lie in the
$\ep$-neighborhood of points  $((z_j)_{j\in I}, (z_j)_{j\in \bar I})$.

This proves all of the statements of Theorem \ref{thm deg} except the last statement that 
 all of the numbers $t_1,\dots,t_m$, $s_1,\dots,s_m$, 
$z_1,\dots,z_{2m}$ are distinct, see the theorem. We already know that
\begin{enumerate}

\item{}
$\Wr(e^{2\pi i kx}\prod_{j=1}^m\theta(x-t_j), e^{-2\pi i(\mu+ k)x}\prod_{j=1}^m\theta(x-s_j))=
e^{-2\pi i \mu x}\prod_{j=1}^{2m}\theta(x-z_j)$,
\item{}
 all  the numbers 
$t_1,\dots,t_m$, $s_1,\dots,s_m$ are distinct,

\item{}
all the numbers
$z_1,\dots,z_{2m}$ are distinct.

\end{enumerate}
 This implies that  $t_1,\dots,t_m$, $s_1,\dots,s_m$, 
$z_1,\dots,z_{2m}$ are distinct. 
Theorem \ref{thm deg} is proved.

\section{Applications of Theorem \ref{thm deg}}
\label{sec appls}

\subsection{Counting ratios of theta-polynomials}

Fix a fundamental parallelogram $\La\subset \C$.
Consider  the ratio $F$ of two theta-polynomials of degree $m$ 
with $m$ simple poles in $\La$.
Then the function $F$ can be  written uniquely in the form 
\bean
\label{hfun}
F=g/f,\qquad 
f={\prod}_{j=1}^m\theta(x-t_j,\tau),  
\eean
where $t=(t_1,\dots,t_m)$ is a point of $\La^m/S_m$ with distinct coordinates
and $g$ is a theta-polynomial of degree $m$. The derivative $F\rq{} = \Wr(f,g)/f^2$
 can be written uniquely in the form
\bean
\label{derh}
F\rq{} = e^{-2\pi i \mu x}\frac{\prod_{a=1}^{2m}\theta(x-z_a,\tau)}{\prod_{j=1}^m\theta(x-t_j,\tau)^2}
\eean
for some $\mu\in\C$ and $z=(z_1,\dots,z_{2m})\in\La^{2m}/S_{2m}$. By assumptions
 the numerator and denominator of this ratio have no common zeros.

\begin{thm}
\label{thm ratio}
Let $z=(z_1,\dots,z_{2m})\in (\La\rq{})^{2m}/S_{2m}$ have distinct coordinates. Then there exists
$N>0$, such that for any $\mu\in \C$ with $|\mu|> N$
there exist exactly $\binom{2m}m$ functions $F(x)$ as in \Ref{hfun} with the derivative as in
\Ref{derh}, up to  proportionality.
\end{thm}

\begin{proof} By construction,
these functions $F=g/f$ are in the bijective correspondence with the points 
$(f,g)\in \Wr^{-1}_0(h)$, where 
$h=e^{-2\pi i \mu x}\prod_{a=1}^{2m}\theta(x-z_a)$. 
 Now Theorem \ref{thm ratio} is a corollary of Theorem \ref{thm deg}, see also
Lemma \ref{lem hol}.
\end{proof}

\subsection{Counting normal solutions of Bethe ansatz equations}
\label{sec couS}

Choose $z=(z_1$, \dots, $z_{2m})\in (\La\rq{})^{2m}/S_{2m}$ with
distinct coordinates.
By Theorem \ref{thm deg} there exists  $N >0$, such that for any 
$h\in P_{2m}-P_{2m,1}$ of the form
\bean
\label{h  mu}
h(x)= e^{-2\pi i\mu x}{\prod}_{a=1}^{2m}\theta(x-z_a,\tau),
\eean
 with $|\on{Im} \mu|>N$,
 the set $\Wr^{-1}_0(h)$ consists of
 exactly $\binom{2m}{m}$ points.
In that case any  $(f,g)\in \Wr^{-1}_0(h)$ has the form
\bean
\label{nfs}
f(x) = {\prod}_{j=1}^m\theta(x-t_j,\tau),\qquad
g(x) = e^{-2\pi i\mu x}{\prod}_{j=1}^m\theta(x-s_j,\tau),
\eean
where $t=(t_1,\dots,t_m)\in(\La\rq{})^m/S_m$, $s=(s_1,\dots,s_m)\in(\La\rq{})^m/S_m$,
and all the numbers $t_1,\dots,t_m$, $s_1,\dots,s_m$, 
$z_1,\dots,z_{2m}$ are pairwise distinct.

In particular, this means that
\bean
\label{mswe}
\phantom{aaaaa}
\Wr({\prod}_{j=1}^m\theta(x-t_j,\tau),
e^{-2\pi i\mu x}{\prod}_{j=1}^m\theta(x-s_j,\tau))=
e^{-2\pi i\mu x}{\prod}_{a=1}^{2m}\theta(x-z_a,\tau))
\eean
and
\bean
\label{msswe}
\phantom{aaaaa}
\Wr(e^{2\pi i\mu x}{\prod}_{j=1}^m\theta(x-t_j,\tau),
{\prod}_{j=1}^m\theta(x-s_j,\tau))=
e^{2\pi i \mu x}{\prod}_{a=1}^{2m}\theta(x-z_a,\tau)).
\eean

By Lemmas \ref{lem:res} - \ref{lem 1.11} equation \Ref{mswe} implies that 
 $t=(t_1,\dots,t_m)$ is a solution of the Bethe ansatz equations
\bean
\label{maBAE}
\phantom{aaa}
2\pi i \mu +2{\sum}_{\ell,\, \ell\ne j}\rho(t_j-t_\ell,\tau) - {\sum}_{a=1}^{2m} \rho(t_j-z_a,\tau) = 0,  \qquad j=1,\dots,m,
\eean
and equation  \Ref{msswe}  implies that 
 $s=(s_1,\dots,s_m)$ is a solution of the Bethe ansatz equations
\bean
\label{mmaBAE}
\phantom{aaa}
-2\pi i \mu +2{\sum}_{\ell,\, \ell\ne j}\rho(s_j-s_\ell,\tau) - {\sum}_{a=1}^{2m} \rho(s_j-z_a,\tau) = 0,  \qquad j=1,\dots,m.
\eean
Hence the equivalence classe of the
solution $(\mu, t_1,\dots,t_m,z_1,\dots,z_{2m})$
and the equivalence class of the solution
$(-\mu,$ $s_1,\dots$, $s_m$, $z_1,\dots,z_{2m})$ belong to the set $\on{Sol}_m$ and
one of them is the image of the other under the analytic involution 
$\beta^{-1} \circ\iota_{\on{an}}\circ\beta$, see \Ref{GenS}.

\smallskip
Given $\mu\in\C$ and  $z=(z_1,\dots,z_{2m})\in \Lambda^{2m}/S_{2m}$
denote by $B(\mu, z)$ the set of equivalence classes of solutions
$(\mu\rq{}, t_1\rq{},\dots,t_m\rq{}, z_1,\dots,z_{2m})$ of the Bethe ansatz equations \Ref{BAEs},
which have representatives of the form
$(\mu, t_1,\dots,t_m, z_1,\dots,z_{2m})$ with $(t_1,\dots,t_m)\in \La^m/S_m$.

\begin{thm} 
\label{cor de}

Let $z=(z_1,\dots,z_{2m})\in (\Lambda\rq{})^{2m}/S_{2m}$ have distinct coordinates. 
Then there exists  $N >0$, such that for any $\mu$ with $|\mu|>N$, each of the sets 
$B(\mu,z)$  and 
$B(-\mu,z)$ consist of  exactly $\binom{2m}{m}$ points. Moreover, there is a bijection
$B(z,\mu)\to B(z,-\mu)$, given by the analytic involution 
$\beta^{-1} \circ\iota_{\on{an}}\circ\beta$, which combines the points of these two sets
 into  $\binom{2m}m$ pairs $(t, s) \in B(\mu,z)\times B(-\mu,z)$ so that
the pairs $(\prod_{j=1}^m \theta(x-t_j), e^{-2\pi i \mu x}\prod_{j=1}^m\theta(x-s_j))$ list all the points
of the set $\Wr^{-1}_0(h)$, see \Ref{mswe}.

\end{thm}

\begin{proof}
The theorem follows from Lemma \ref{lem eqv} and Theorem \ref{thm deg}, 
see also Lemma \ref{lem hol}.
\end{proof}

\subsection{Asymptotics of solutions of Bethe ansatz equations}
\label{sec aobe}

Let $z=(z_1,\dots,z_{2m})\in (\La\rq{})^{2m}/S_{2m}$ have distinct coordinates.
By Lemma \ref{lem hol} the Bethe ansats equations \Ref{BAEs} extend to a holomorphic system of equations for $(\mu, t) \in \P^1_\C \times \C^m$ at $\mu=\infty$.

\begin{lem}
When $\mu = \infty$, a point $(t_1,\dots,t_m)\in \La^{m}/S_m$
 is a solution of \Ref{BAEs} if and only if $\{t_1,\dots,t_m\}$ is a subset of size $m$ of $\{z_1,\dots,z_{2m}\}$.
\qed
\end{lem}

Thus, solutions to \Ref{BAEs} at $\mu = \infty$ for $t \in \La^m/S_m$
are in bijection with $m$-element subsets $I\subset\{1,\dots,2m\}$.  By the implicit function theorem, each of these extends to a holomorphic family of solution $t_I(\mu) \in \La^m/S_m$ to \Ref{BAEs} as $\mu$ ranges through a neighborhood of $\infty$,
see the proof of Theorem \ref{thm deg}.

Thus, $t_{I}(\mu)$ 
is a holomorphic function of $\mu$ as $\mu \to\infty$ and 
 has the property:
\bean
\label{limts}
t_{I}(\mu)\to (z_j)_{j\in I},\quad
\on{as}\quad  \mu \to\infty.
\eean

\begin{lem}
\label{lem 1.23}

For $I=\{i_1,\dots,i_m\}$,  the coordinates 
$t_1(\mu),\dots,t_m(\mu)$ of the solution 
 $t_{I}(\mu)$ can be ordered so that for $\mu \to\infty$ we have
\bean
\label{limco}
&&
t_j(\mu) = z_{i_j} + \frac 1{2\pi i \mu} + O(\mu^{-2}), 
\eean
where
$j=1,\dots,m$.
\end{lem}

\begin{proof}
Denote $v=(2\pi i \mu)^{-1}$ as in  Lemma \ref{lem hol}. 
Introduce new variables $w_1,\dots,w_m$ by the formula
$t_j=z_{i_j} + v w_j$. Then equations \Ref{BAEs} can be written in the form
\bean
\label{appr}
\frac 1v -\rho(vw_j) + R_j(v, w_1,\dots,w_m) = 0, \qquad j=1,\dots, m,
\eean
where $R_j(v, w_1,\dots,w_m)$ is a holomorphic function of its arguments at
$v=0$. Since $\rho(u)=u^{-1}+O(u)$ as $u\to 0$, equation \Ref{appr} can be rewritten
in the form
$\frac 1v -\frac1{vw_j} + \tilde R_j(v, w_1,$ \dots, $w_m) = 0$, where
$\tilde R_j(v, w_1,\dots,w_m)$ is another holomorphic function. Multiplying both sides of this equation
by $vw_j$ we obtain an equation $w_j = 1 - v \tilde R_j(v, w_1,\dots,w_m)$, which implies that 
$w_j =1 + O(v)$. This proves   \Ref{limco}.
\end{proof}

As $\mu\to \infty$, the analytic involution $\beta^{-1} \circ\iota_{\on{an}}\circ\beta$
sends the equivalence class of the solution $(\mu, t_I(\mu))$ to the equivalence class of the solution
$(-\mu, t_{\bar I}(-\mu))$, where $\bar I=\{1,\dots,2m\}-I$, see
Lemma \ref{lem hol} and Theorem \ref{cor de}.

\begin{cor}
\label{cor sa} 

For $\bar I=\{\bar i_1,\dots, \bar i_m\}$, the coordinates
$t_1(-\mu),\dots,t_m(-\mu)$ of the solution 
 $t_{\bar I}(-\mu)$ can be ordered so that for $\mu \to\infty$ we have
\bean
\label{lim-}
&&
t_j(-\mu) = z_{\bar i_j} - \frac 1{2\pi i \mu} + O(\mu^{-2}), 
\eean
where
$j=1,\dots,m$.
\qed

\end{cor}

\subsection{Asymptotics of Bethe eigenfunctions}

Let $\mu\to \infty$. Choose one of the $\binom{2m}{m}$ 
families of Bethe ansatz solutions $t_I(\mu)\in \La^m/S_m$,
where $I=\{i_1,\dots,i_m\}$ is an $m$-element subset of $\{1,\dots,2m\}$.
Let us order the coordinates $t_1(\mu), \dots,t_m(\mu)$ so that the asymptotics 
\Ref{limco} hold.
Consider the eigenfunction $\Psi(\la_{12},\mu, t_I(\mu),z)$ corresponding to this solution.
Then the function
\bean
\label{neorme}
\Psi_I(\la_{12},\mu):= \Psi(\la_{12},\mu, t_I(\mu),z) \,
{\prod}_{j=1}^m \theta(t_j(\mu)-z_{i_j})
\eean
is also an eigenfunction of the dynamical elliptic Bethe algebra,
 since the last product
 is constant with respect to $\la_{12}$. The function $\Psi_I(\la_{12},\mu)$
 will be called the {\it normalized eigenfunction} corresponding to the eigenfunction 
$\Psi(\la_{12},\mu, t_I(\mu),z)$.

\begin{lem}
\label{lem asPsi}
The normalized function $\Psi_I(\la_{12},\mu)$  has an expansion of the form
\bean
\label{as Psi}
\Psi_I(\la_{12},\mu) = 
e^{\pi \mu\la_{12}}  \sum_{k=0}^\infty w_k(\la_{12}) \,\mu^{-k},
\eean
where $w_0=(-1)^m v_{\bar I}$ does not depend on $\la_{12}$, all the coefficients
$w_k(\la_{12})$ are meromorphic functions
of $\la_{12}$ with poses at most at the points of the subset $\Z+\tau\Z\subset \C$,
 and the sum is 
uniformly convergent on any compact subset of $\C-(\Z+\tau\Z) $ in the $\la_{12}$-line.

\end{lem}

\begin{proof}
Clearly the function
$F(\la_{12},\mu)=e^{-\pi \mu\la_{12}} \theta(\la)^m\Psi_I(\la_{12},\mu)$
is holomorphic in $(\la_{12}, \mu)$ in a neighborhood of the line $\mu=\infty$,  moreover,
$F(\la_{12},\infty) = (-1)^m\theta(\la)^mv_{\bar I}$, see formula \Ref{prod}.
 This implies the lemma.
\end{proof}

Consider the eigenvalues  $E_{0,I}(\mu), \dots, E_{2m ,I}(\mu) $
of the dynamical Hamiltonians $H_0(z)$, \dots, $H_{2m}(z)$ 
on the eigenfunction $\Psi_I(\la_{12},\mu)$.

\begin{lem}
\label{lem aeig}

As $\mu\to\infty$ we have
\bean
E_{0,I} (\mu) 
&=& \frac{\pi i}2\mu^2 + {\sum}_{k=0}^\infty E_{0,I}^k \mu^{-k}
\\
\notag
E_{a,I} (\mu) 
&=& \pi i \mu + {\sum}_{k=0}^\infty E_{a,I}^k \mu^{-k}, \qquad a\in I,
\\
\notag
E_{a,I} (\mu) 
&=& -\pi i \mu + {\sum}_{k=0}^\infty E_{a,I}^k \mu^{-k}, \qquad a\not\in I,
\eean
where for $a=0,\dots,2m$, the functions 
$\sum_{k=0}^\infty  E_{a,I}^k \mu^{-k}$ 
are holomorphic functions at $\mu=\infty$.

\end{lem}

\begin{proof} The eigenvalues of the KZB operators on the eigenfunction 
 $\Psi(\la_{12},\mu, t_I(\mu),z)$ were calculated in Section \ref{sec Fdo}.
Now the lemma follows from formula \Ref{limco}.
\end{proof}

\begin{cor}
\label{cor assol}
The  formal series $e^{\pi \mu\la_{12}}  \sum_{k=0}^\infty w_k(\la_{12}) \,\mu^{-k}$
with respect to the variable $\mu$ 
is a solution of the equations 
\bean
\label{assol}
\phantom{aaa}
H_0(z)  \sum_{k=0}^\infty w_k(\la_{12}) \,\mu^{-k}
&=&
 \big(\frac{\pi i}2\mu^2 + {\sum}_{k=0}^\infty E_{0,I}^k \mu^{-k}\big)
\sum_{k=0}^\infty w_k(\la_{12}) \,\mu^{-k},
\\
\notag
H_a(z)  \sum_{k=0}^\infty w_k(\la_{12}) \,\mu^{-k}
&=&
 (\pi i \mu + {\sum}_{k=0}^\infty E_{a,I}^k \mu^{-k})
\sum_{k=0}^\infty w_k(\la_{12}) \,\mu^{-k},
\quad a\in I,
\\
\notag
H_a(z)  \sum_{k=0}^\infty w_k(\la_{12}) \,\mu^{-k}
&=&
 (-\pi i \mu + {\sum}_{k=0}^\infty E_{a,I}^k \mu^{-k})
\sum_{k=0}^\infty w_k(\la_{12}) \,\mu^{-k},
\quad a\not\in I.
\eean
\qed
\end{cor}

\subsection{Asymptotic eigenfunctions}
\label{sec aseign}

Consider a formal series 
\bea
\Psi^{\on{asy}}(\la_{12},\mu) = e^{2\pi i \mu\la_{12}}\sum_{k=0}^\infty u_k(\la_{12}) \,\mu^{-k},
\qquad
\mu\to\infty
\eea
with some coefficients $u_k(\la_{12})$.

\begin{lem}
\label{lem fass} 

If there exists a formal series solution
\bea
\Psi^{\on{asy}} (\la_{12},\mu) =  e^{2\pi i \mu\la_{12}}
{\sum}_{k=0}^\infty u_k(\la_{12}) \,\mu^{-k}
\eea
to the equations
\bean
\label{fassol}
\phantom{aaa}
H_0(z) \Psi^{\on{asy}} (\la_{12},\mu)
&=&
 \big(\frac{\pi i}2\mu^2 + {\sum}_{k=0}^\infty E_{0,I}^k \mu^{-k}\big)
\Psi^{\on{asy}} (\la_{12},\mu),
\\
\notag
H_a(z)  \Psi^{\on{asy}} (\la_{12},\mu)
&=&
 (\pi i \mu + {\sum}_{k=0}^\infty E_{a,I}^k \mu^{-k})
\Psi^{\on{asy}} (\la_{12},\mu),
\quad a\in I,
\\
\notag
H_a(z)  \Psi^{\on{asy}} (\la_{12},\mu)
&=&
 (-\pi i \mu + {\sum}_{k=0}^\infty E_{a,I}^k \mu^{-k})
\Psi^{\on{asy}} (\la_{12},\mu),
\quad a\not\in I,
\eean
then it is unique up to multiplication by a scalar of the form 
$\sum_{k=0}^\infty c_k\mu^{-k}$, where $c_i\in\C$.

\end{lem}

\begin{proof}  Equations \Ref{fassol} say that
\bean
\label{in0}
u_0\rq{} 
&=& 0,
\\
\notag
((e_{22}-e_{11})^{(a)} -1) u_0 &=& 0,      \qquad a\in I,
\\
\notag
((e_{22}-e_{11})^{(a)} +1) u_0 &=& 0,      \qquad a\in \bar I,
\eean
where $u\rq{}$ is the derivative of the function $u$ of the argument $\la_{12}$ with respect to
$\la_{12}$.  Hence $u_0$ is a constant multiple of the vector $v_{\bar I}$, that is, 
$u_0 = c_0 v_{\bar I}$ for some $c_0\in\C$. For $k>0$,  equations \Ref{fassol} say that
\bean
\label{ink}
u_k\rq{} 
&=&  R_k^0(u_{0},\dots,u_{k-1}) ,
\\
\notag
((e_{22}-e_{11})^{(a)} -1) u_k &=&  R_k^a(u_{0},\dots,u_{k-1})  ,      \qquad a\in I,
\\
\notag
((e_{22}-e_{11}) +1) u_k &=&  R_k^a(u_{0},\dots,u_{k-1}),      \qquad a\in \bar I,
\eean
where $R_k^a(u_{0}$, \dots, $u_{k-1}) $, for $a=0,\dots,2m$, are some explicit expressions in terms of 
the functions $u_{0}$, \dots, $u_{k-1}$.
So if the coefficient $u_k$ can be determined from the system \Ref{ink}, then it is unique up 
to addition
of a constant multiple of $v_{\bar I}$. This proves the lemma.
\end{proof}

\subsection{Proof of Theorem \ref{thm inv}}
\label{end}

Let $I\subset\{1,\dots,2m\}$ be an $m$-element subset and $\mu\to\infty$. 
Consider the solutions $(\mu,t_I(\mu))$ and $(-\mu, t_{\bar I}(-\mu))$ 
of the Bethe ansatz equations \Ref{BAEs}. They are related by the analytic involution
$\beta^{-1} \circ\iota_{\on{an}}\circ\beta$, see Section \ref{sec aobe}.
Consider the associated eigenfunctions $ \Psi(\la_{12},\mu, t_I(\mu),z)$
and  $\Psi(\la_{12},-\mu, t_{\bar I}(-\mu),z)$ and their respective
normalized versions, which we will denote by
$\Psi_I$ and $\Psi_{\bar I}$,
respectively,  see \Ref{neorme}. Both have asymptotic expansions of Lemma \ref{lem asPsi},
\bea
\Psi_I = 
e^{\pi \mu\la_{12}}  {\sum}_{k=0}^\infty w^I_k(\la_{12}) \,\mu^{-k},
\qquad
\Psi_{\bar I} = 
e^{-\pi \mu\la_{12}}  {\sum}_{k=0}^\infty \bar w^{\bar I}_k(\la_{12}) \,\mu^{-k},
\eea
where $w^I_0=(-1)^m v_{\bar I}$ and $\bar w^{\bar I}_0=(-1)^mv_{I}$.
Recall the nontrivial element $s$ of the $\slt$ Weyl group and consider the third eigenfunction
$s(\Psi_I)$. Its asymptotic expansion has the form
$e^{-\pi \mu\la_{12}}  \sum_{k=0}^\infty\, s.w^{ I}_k(-\la_{12}) \,\mu^{-k}$,
where $s.w_0^I(-\la_{12}) = (-1)^mv_I$. By Lemma \ref{lem fass} we conclude that 
$s(\Psi_I) = \Psi_{\bar I}$. Hence Theorem \ref{thm inv} is proved for solutions of the form
$(\mu,t_I(\mu))$. By Theorem \ref{thm deg} such solutions correspond to an open subset of the
space $\on{Sol}_m$, which is irreducible. This proves Theorem \ref{thm deg}
in full generality.
\qed

\end{document}